\documentclass[11pt]{article}
\usepackage[english]{babel}
\usepackage[utf8]{inputenc}
\usepackage{fancyhdr}
\usepackage{amssymb,amsfonts,amsthm,amscd}
\usepackage{amsmath}

\usepackage{graphicx}
\usepackage{psfrag}

\newcommand{\bm}[1]{\mbox{\boldmath $#1$}}

\DeclareFontFamily{OT1}{rsfs}{}
\DeclareFontShape{OT1}{rsfs}{m}{n}{ <-7> rsfs5 <7-10> rsfs7 <10->
rsfs10}{} \DeclareMathAlphabet{\mycal}{OT1}{rsfs}{m}{n}

\newcommand{\mnotex}[1]
{\protect{\stepcounter{mnotecount}}$^{\mbox{\footnotesize $\bullet$\themnotecount}}$ 
\marginpar{
\raggedright\tiny\em
$\!\!\!\!\!\!\,\bullet $\themnotecount: #1} }

\newtheorem{Tma}{Theorem}
\newtheorem{Prop}{Proposition}
\newtheorem{Def}{Definition}
\newtheorem{Lem}{Lemma}
\newtheorem{Cor}{Corollary}
\newtheorem{remark}{Remark}
\theoremstyle{remark} 

\def\Kkzero{K^k_0}

\def\defi{:=}
\def\scri{{\mycal I}}

\def\Scalgam{\mathrm{Scal}^{\gamSig}}
\def\Eing{\mathrm{Ein}^g}
\def\Ricg{\mathrm{Ric}^g}

\def\Riemg{\mathrm{Riem}^g}
\def\tr{\mathrm{tr}}
\def\N{\Omega}

\def\thk{\theta_k}
\def\thl{\theta_{\ell}}
\def\thlvarphi{\theta_{\ellvarphi}}
\def\Sig{S}
\def\la{\langle}
\def\ra{\rangle}
\def\ellvarphi{\ell^{\varphi}}
\def\Dlam{D_{\lambda}}

\def\s{s}
\def\mm{M(\Siglam,\ell)}
\def\mmvarphi{M(\Sig,\ell^{\varphi})}
\def\mmlam{M(\Sig_{\lambda},\ell)}
\def\mmlamp{M(\Sig_{\lambda},\ell^\star)}

\def\D{D}
\def\G{G}
\def\Siglam{\Sig_{\lambda}}
\def\Sigmu{\Sig_{\mu}}
\def\gamSiglam{\gamma_{\Siglam}}
\def\gamSigmu{\gamma_{\Sigmu}}
\def\gamSigzero{\gamma_{\Sig_0}}
\def\volSiglam{\bm{\eta_{\Siglam}}}
\def\volSigzero{\bm{\eta_{\Sig_0}}}
\def\volunitdos{\bm{\eta_{\q}}}

\def\M{{\mathcal M}}
\def \Mcal {\mathcal{M}^{1,3}}
\def\gMm{g}

\def\gamSig{\gamma_{\Sig}}
\def\volSig{\bm{\eta_{\Sig}}}
\def\qh{\hat{q}}
\def\q{{\mathring{q}}}
\def\gammatilde{\tilde{\gamma}}
\def\s{s}
\def\sone{\s^{(1)}_{\ell}{}}
\def\stilde{\tilde{\s}}

\def\thlone{\theta_{\ell}^{(1)}}
\def\thltilde{\tilde{\theta}_{\ell}}
\def\thkone{\theta_{k}^{(1)}}
\def\thktwo{\theta_{k}^{(2)}}
\def \esfdos{\mathbb{S}^2}
\def \rq{R_{\hat{q}}}
\def\Sh{\hat{S}}
\def\volsh {\bm{\eta_{\hat{q}}}} 
\def\ct {c} 
\def \f {{f}}

\def\esf {\mathring{\nabla}}

\def\Scal{\mathrm{Scal}}

\def\K{\mathcal K}
\def\c {{\thkone}}
\def \cdos {{\thktwo}}
\def\a {{\thlone}}
\def\Id{\mathrm{\bm{Id}}}
\def\A{\thlone}
\def\del {\bm{\delta}}
\def\JournalPrep#1#2#3{#1, ``#2'', #3.}
\def\Journal#1#2#3#4#5#6{#1, ``#2'', {\em #3} {\bf #4}, #5 (#6).}

\def\JDG{\em J. Diff. Geom.}
\def\CQG{\em Class. Quantum Grav.}
\def\JPA{\em J. Phys. A: Math. Gen.}
\def\PRD{{\em Phys. Rev.} {\bf D}}

\def\CMP{\em Commun. Math. Phys.}

\def\PREP{\em Phys. Rep.}

\def\ANYAS{\em Ann. N. Y. Acad. Sci.}

\def\AM{\em Ann. Math.}

\def\DMJ{\em Duke Math. J.}

\begin{document}

\title{On the Penrose inequality along null hypersurfaces}

\author{Marc Mars$^1$ and Alberto Soria$^2$ \\
Facultad de Ciencias, Universidad de Salamanca,\\
Plaza de la Merced s/n, 37008 Salamanca, Spain \\
$^1$ \,marc@usal.es,  $^2$ \,asoriam@usal.es, }

\maketitle

\begin{abstract}
The null Penrose inequality, i.e. the Penrose inequality in terms
of the Bondi energy, is studied by introducing a funtional
on surfaces and studying its properties along a null
hypersurface $\Omega$ extending to past null infinity. We prove a general 
Penrose-type inequality which involves
the limit at infinity of the Hawking energy along a
specific class of geodesic foliations called {\it Geodesic Asymptotic Bondi}
(GAB), which are shown to always exist.
Whenever, this foliation approaches large  spheres, this inequality
becomes the null Penrose inequality and we recover the results
of Ludvigsen-Vickers and Bergqvist. By exploiting further properties of 
the functional along general geodesic foliations, we introduce
an approach to the null Penrose inequality
called {\it Renormalized Area Method} and find a set of
two conditions which implies the validity of the
null Penrose inequality. One of the conditions involves a limit at infinity and
the other a condition on the spacetime curvature along the flow. We investigate
their range of applicability in two particular but interesting cases, namely
the shear-free and vacuum case, where the null Penrose inequality
is known to hold from the results by Sauter, and the case of null
shells propagating in the Minkowski spacetime. Finally, a general inequality
bounding the area of the quasi-local black hole in terms of an asymptotic
quantity intrinsic of $\Omega$ is derived.

\end{abstract}

\section{Introduction}

The Penrose inequality in asymptotically flat spacetimes satisfying
the dominant energy condition conjectures
that the total energy measured by any observer is bounded below in terms
of the area of suitable spacelike surfaces related to quasi-local black holes.
This conjecture has received much attention since its formulation in 
\cite{Penrose1973} and constitutes an important open problem in gravitation.
It has been proved in full generality only in spherical symmetry 
\cite{MalecMurchadha1994, Hayward1996}
and for time symmetric hypersurfaces \cite{HuiskenIlmanen2001},
\cite{Bray2001} (extended to space dimensions up to seven in \cite{BrayLee}). 
For a review of the results prior to 2009
the reader is referred to \cite{Mars2009}.

In recent years one of the several lines of research that have been
pursued involves the Penrose inequality in the null case. Here, the
total energy of the spacetime is the Bondi energy $E_B$ measured
by an asymptotically inertial observer at a cut $S_{\infty}$
of past null infinity $\scri^{-}$ and the quasi-local black hole is
a spacelike surface $S_0$ with the two properties of (i) having
non-positive future outer null expansion (i.e.
it is a {\it weakly outer trapped surface} )
and (ii) the outgoing, past directed
null hypersurface $\Omega$ starting at $S_0$ extends smoothly
all the way to infinity and intersects $\scri^{-}$ at $S_{\infty}$.
Since, in such a setup, $S_0$ has smaller area than any
other surface embedded in $\Omega$ to the past of $S_0$, the Penrose inequality has the form
\begin{equation}
\label{PIa}
E_B \geq \sqrt{\frac{|S_0|}{16\pi}}
\end{equation}
and does not require invoking minimal area enclosures of the quasi-local
black hole surface $S_0$, as in the general case.
This version of the Penrose inequality is often referred
to as the {\it null Penrose inequality}. It has been proved only in a few
special cases, including the case when $\Omega$ is shear-free
and vacuum by Sauter \cite{Sauter2008}. Using a non-linear
perturbation argument around a spherically symmetric null hypersurface
in the Schwarzschild spacetime (which is indeed shear-free), Alexakis 
has been able
to prove \cite{Alexakis}
the null Penrose inequality for vacuum spacetimes
close enough (in a suitable sense) to the Schwarzschild exterior 
spacetime.  The null Penrose inequality contains, as
a particular case, the original formulation due to Penrose involving
null shells of dust propagating in the Minkowski spacetime. This problem
has also received attention recently, both for shells in Minkowski
\cite{Wang1}, \cite{MarsSoria2012} as well as for a related conjecture in the
Schwarzschild spacetime  \cite{BrendleWang}.

A proof of the general null Penrose inequality was claimed 
by Ludvigsen and Vickers \cite{LudvigsenVickers1983}. However,
a gap was found by Bergqvist \cite{Bergqvist1997}
who, at the same time,
substantially streamlined the argument. Since 
Ludvigsen-Vickers \& Bergqvist's argument is relevant for this paper
let us describe it in some detail. Their method was based on two facts. The
first one was the existence of a quasi-local  object defined
on surfaces which enjoyed monotonicity properties along past directed
null geodesic foliations. This functional was introduced by
Bergqvist \cite{Bergqvist1997} ant it has been called
sometimes Bergqvist mass in the literature \cite{Mars2009},
\cite{MarsSoria2012}. The second fact was a suitable upper bound for the
area of the weakly outer trapped surface $S_0$.  Establishing  this bound
involved that the geodesic null foliation $\{ S_{r} \}$ of $\Omega$
starting at $S_{r_0} = S_0$ (where $r_0 \in \mathbb{R}^+$ and
the range of $r$ is $[r_0,\infty)$)
dragging $S_0$ to 
past null infinity satisfied two additional properties. The first
one was that 
the future null expansion $\theta_k$ of $S_{r}$
along the future null generator $k$ tangent to $\Omega$ admits
an expansion of the form
\begin{equation}
\label{LVBexpansion}
\theta_k=\frac{-2}{r}+O\left(\frac{1}{r^3}\right), 
\end{equation}   
i.e. with vanishing coefficient in the term $r^{-2}$. The second
one was that the rescaled metric $r^{-2} \gamma(r)$ (where 
$\gamma(r)$ is the induced metric of $S_r$) approaches a round metric
on the sphere when $r \rightarrow + \infty$ (one says that
$\{ S_{r} \}$ approaches large spheres). The main result by
Ludvigsen-Vickers is that under these circumstances the 
Penrose inequality (\ref{PIa}) follows. Lugvigsen and Vickers took for
granted that a geodesic foliation $\{ S_r \}$ satisfying these two
properties always exists. Bergqvist noted that under the assumption
(\ref{LVBexpansion}) it was not at all clear that the condition
that the metric $r^{-2} \gamma(r)$ approaches a round sphere needs to be 
satisfied. This was the gap in 
the original paper \cite{LudvigsenVickers1983}.
In \cite{MarsSoria2012} we investigated the Penrose
inequality for dust null shells in Minkowski and proved its validity
for a large class of surfaces. It turns out that the class of 
surfaces $S_0$ for which the Ludvigsen-Vickers method  applies
is very restrictive \cite{MarsSoriaProc}. On any given past directed
outward null hypersurface $\Omega$ in the Minkowski spacetime extending
smoothly all the way to past null infinity, there was only a one-parameter
family of surfaces for which the Ludvigsen-Vickers \& Bergqvist argument applies. Our method
in \cite{MarsSoria2012} 
was based on geodesic foliations approaching large
spheres but did not rely on the condition (\ref{LVBexpansion}). 
The arguments however were tailored to the Minkowski spacetime where
the null dust shell propagates. It makes sense to try and
extend the ideas of \cite{MarsSoria2012}, which in turn
were motivated by Bergqvist's approach, to find sufficient conditions
for the null Penrose
inequality in general asymptotically flat spacetimes satisfying
the dominant energy condition. This is one of the main objectives
of the present paper.

The second main objective is complementary to the previous one. Instead
of relaxing condition (\ref{LVBexpansion}) and keeping the assumption
that $\{ S_{r} \}$ approaches large spheres, it is natural
to consider the setup 
when (\ref{LVBexpansion})
is kept and we relax the condition of approaching large
spheres. Geodesic foliations with this property are
named ``Geodesic Asymptotically Bondi'' in this paper (or GAB
for short). A motivation for this name will be given later.
GAB foliations turn out to always exists and be (geometrically)
unique given any cross section $S_0$ in a past asymptotically
flat null hypersurface. Our main result in this setting is a
Penrose type inequality which relates the
area of any weakly outer trapped surface $S_0$ and the limit at infinity of the Hawking energy along the GAB foliation
associated to $S_0$. More precisely (see below for the precise
definitions).
\begin{Tma}[\bf A Penrose type inequality for GAB foliations]
\label{main1}
Let $\Omega$ be a past asymptotically flat null hypersurface
in a spacetime $(\M,g)$ satisfying the dominant energy condition. Let
$S_0$ be a spacelike cross section of $\Omega$. If $S_0$ is a
weakly outer trapped surface, then
\begin{equation*}
\sqrt{\frac{|S_0|}{16\pi}} \leq 
\underset{\lambda \to \infty}{\lim}  m_H(S_{\lambda}), 
\end{equation*} 
where $m_H(S)$ denotes the Hawking energy of $S$ and
$\{ S_{\lambda} \}$ is the GAB foliation associated to $S_0$.
\end{Tma}

In combination with a study of the limit of the
Hawking energy along general foliations $\{ S_{\lambda} \}$
of asymptotically flat null hypersurfaces $\Omega$ carried
out in \cite{MarsSoria2015}, this theorem provides
an interesting
Penrose-type inequality with potentially useful applications. 
This theorem immediately extends Ludvigsen-Vickers \& Bergqvist result
because when $\{ S_{\lambda} \}$ approaches large spheres one automatically
has $m_H(S_{\lambda}) \longrightarrow E_B$, where
$E_B$ is the Bondi energy at the cut at $\scri^{-}$ defined
by $\Omega$ and measured by the observer defined by
$\{ S_{\lambda} \}$.

The key object in this paper is the functional on surfaces 
\begin{equation*}
M(S,\ell)=\sqrt{\frac{|S|}{16\pi}}-\frac{1}{16\pi}\int_{S}\theta_{\ell} \bm{\eta_{S}},
\end{equation*}
which has the property that
\begin{equation*}
\sqrt{\frac{|S|}{16\pi}}\leq  M(S,\ell)
\end{equation*}
whenever $S$ is a weakly outer trapped surface. It also has the property
that its limit at infinity along foliations approaching large spheres
is the Bondi energy.
The main objective of 
this paper is to bound $M(S,\ell)$ from above by its limit at infinity.
For that, the
monotonicity properties of $M(S_{\lambda},\ell)$ along suitable foliations will be
studied. Although in general, this object is not monotonic, it can be split
in two pieces, $M_b(S_{\lambda},\ell)$ and $D(S_{\lambda},\ell)$, where the
first one is closely related to an object first introduced by
Bergqvist in \cite{Bergqvist1997} and turns out to be monotonically
increasing provided the dominant energy condition holds. Thus, discussing
under which conditions $D(S_{\lambda},\ell)$ is bounded above by its limit 
becomes a problem of interest. We consider various approaches
to such an inequality and analyze their range of applicability
by applying them to two particular but relevant cases, namely
the case when $\Omega$ is shear-free and vacuum
(where, as mentioned, the null Penrose inequality is known to hold 
by other methods \cite{Sauter2008})
and the case of null shells propagating in the Minkowski spacetime. 
The latter
will allow us in particular to provide a link between the 
analysis here and the one in \cite{MarsSoria2012}.

This paper is organized as follows. In Section 2,
after introducing our terminology, we define
the functional of surfaces $M(S,\ell)$ and study its
monotonicity properties, as well as its limit at infinity. For the limit
we use the notion of 
{\bf past asymptotically flat} null hypersurface, introduced
in \cite{MarsSoria2015} to study the limit of the Hawking energy
$m_H(S)$ at infinity. This allows us to relate the limits of $M(S,\ell)$ 
and $m_H(S)$ along geodesic foliations.
In Section 3 we introduce the concept of
{\bf Geodesic Asymptotic Bondi (GAB) foliation},
and study its existence and uniqueness properties. We split 
$M(S,\ell)$ into $M_b(S,\ell)$ and $D(S,\ell)$
and prove that for GAB foliations $D(S_{\lambda},\ell)$ is bounded above by its limit 
at infinity, from which Theorem \ref{main1}  follows. The upper bound of
$D(S_{\lambda},\ell)$ is obtained by studying the monotonicity properties of yet
another functional $F(S_{\lambda})$.  In Section 4 we investigate various sufficient
conditions implying the property that $D(S_{\lambda},\ell)$
is bounded above by its limit, and hence the Penrose inequality along null
hypersurfaces. For this, a slightly stronger notion of 
asymptotic flatness will be required.
We first try to generalize the  method valid
for GAB foliations to more general settings and discuss the difficulties
that arise. We then concentrate on the so-called {\bf 
Renormalized Area Method}, where the null Penrose inequality is
approached via studying the monotonicity properties of
$D(S_{\lambda},\ell)$ itself. The main result here is Theorem 
\ref{penroseineqtma}
where two conditions are spelled out from which the null Penrose
inequality follows. Section 5 is devoted to studying
the shear-free vacuum case as an interesting test bed for the previous ideas.
After showing that one of the two conditions in Theorem 
\ref{penroseineqtma} fails to hold, we nevertheless find an argument
proving the null Penrose inequality using the properties
of $M(S_{\lambda},\ell)$. This not only provides an alternative proof of
Sauter's theorem, but also yields an explicit formula for the Bondi
energy in terms of the geometry of any chosen cross section $S_0$
of the null hypersurface. Section 6 is devoted to studying 
the renormalized area method 
in the Minkowski spacetime. This allows us to recover the results in 
\cite{MarsSoria2012} in a much more direct and efficient way. Concerning the
application of Theorem \ref{main1} to the Minkowski setting, we derive 
a general inequality (Theorem  \ref{IneqMink})
valid for any closed spacelike surface 
in Minkowski for which its outer past null cone extends smoothly to 
past null infinity. In the final section, we quit the 
method involving $M(S,\ell)$ and exploit some results derived
along the way to show a general inequality
bounding the area of a closed spacelike surface embedded
in a past asymptotically flat null hypersurface $\Omega$ in terms of an
asymptotic quantity intrinsic to $\Omega$.

\section{A functional on two-surfaces}

Let ($\mathcal{M},g$) be a time-oriented spacetime of dimension four.
Given a closed (i.e. compact and without boundary) orientable, spacelike,
codimension-two surface $S$ in $(\mathcal{M},g)$, its normal bundle
$NS$  admits a global basis of future directed null vectors $k$ and $\ell$.
The second fundamental form is $\vec{K}(X,Y)=-(\nabla_{X}Y)^\bot$, where $X$ and $Y$ are tangent vectors to $S$,
and $\nabla$ is the covariant derivative in $(\mathcal{M},g)$. 
The null curvatures $K^{k}(X,Y)$ and $K^{\ell}(X,Y)$ are defined by $K^{k}(X,Y)=\langle k,
K(X,Y)\rangle$ (and similarly for $\ell$) and the null-expansions, denoted  by $\theta_{k}$ 
and $\theta_{\ell}$, are the traces of the null curvatures with respect to the induced metric $\gamma$. A key object in this paper is the following functional on $S$
\begin{align}
M(S,\ell) = \sqrt{\frac{|\Sig|}{16 \pi}} - \frac{1}{16 \pi} \int_{\Sig} \theta_{\ell} \volSig
\label{functional}
\end{align}
where $|\Sig|$ is the area of $\Sig$ and $\volSig$ the metric volume form
of $\Sig$. This quantity has geometric units of length so one may be tempted
to assign to it a physical interpretation of quasi-local mass of $S$. However, $M(S,\ell)$
is not  truly a quasi-local quantity on the surface because
it depends on the choice of null normal $\ell$, which cannot be uniquely fixed a priori in the absence
of additional geometric structure. Note, however, that a weakly outer trapped surface
$\Sig_0$  satisfies, by definition, $\theta_{\ell} \leq 0$ irrespectively of the scaling of $\ell$, and hence
\begin{align*}
\sqrt{\frac{|\Sig|}{16 \pi}} \leq M(S,\ell).
\end{align*}
So, if $M(S,\ell)$ enjoyed good monotonicity properties under suitable flows
and its value on very large surfaces in an asymptotically flat
context could be related to the total mass of the spacetime, this object
would be potentially useful to address the Penrose inequality and play perhaps
a similar role as the Hawking energy does in the time-symmetric context. 

It turns out that for null flows $M(S,\ell)$ satisfies an interesting evolution equation. In
order to describe it, let
$\Omega$ be a smooth, connected null hypersurface embedded in ($\mathcal{M},g$) with null
normal $k$ and 
admitting a global cross section $S_0$
(i.e. a smooth embedded spacelike
surface intersected precisely once by every inextendible curve along the 
null generators tangent to $k$).  We want to investigate the derivative of $M(S_{\mu},\ell)$ with respect to $\mu$,
where $S_{\mu}$ is a foliation of $\Omega$ by cross sections. In order to maintain the generality
we do not make any assumption on the null generator $k$ satisfying $k(\mu) = -1$ (other
than being nowhere zero) or on the choice of null normal $\ell$ to $S_{\mu}$ (other
than being transverse to $S_{\mu}$). 
In order to compute the derivative of $M(S_{\mu},\ell)$ we need the following well-known
identities: let $\gamSig(\mu)$ be the induced metric of $S_{\mu}$ and define $Q_k: \Omega \mapsto
\mathbb{R}$ by 
\begin{equation*}
\nabla_k k = Q_{k} k
\end{equation*}
so that $Q_k$ vanishes if and only
if the null geodesic generator $k$ is chosen to be affinely
parametrized. Given a smooth positive function $\varphi: \Omega \mapsto \mathbb{R}^+$,
there is a unique choice of null normal $\ell$ to $S_{\mu}$ (denoted by 
$\ellvarphi$) satisfying 
\begin{equation*}
\la k,\ellvarphi \ra = -\varphi.
\end{equation*}
The choice $\varphi =2$ will be relevant later and we will denote
$\ell^{\varphi=2}$ simply by $\ell$ from now on.
Let us decompose $K^k$ into its trace and trace-free part as $K^k_{AB}=\frac{1}{2}\theta_k\gamma_{AB}+\Pi^k_{AB}$. 
The following evolution equations are standard, see e.g. \cite{GourgoulhonJaramillo}
\begin{align}
k (\gamSig ) & = 2 K^k \label{Evolmetric} \\
k (\thk) = & Q_{k} \thk -\frac{1}{2}{\thk}^2 - \Pi^k_{AB} \Pi^{k\, AB} - \Ricg (k,k) \label{Ray} \\
k (\thlvarphi ) = & \left ( \frac{1}{\varphi} k (\varphi) - Q_{k} \right ) \thlvarphi +
\Eing (k,\ellvarphi ) - \frac{\varphi}{2} \left ( \Scalgam - \la \vec{H}, \vec{H} \ra
\right ) 
+ \varphi \left ( -\mbox{div}_{\Sig} \s_{\ellvarphi}  + |\s_{\ellvarphi}|^2_{\gamSig} \right ) 
\label{kthl2}
\end{align}
where $\Ricg$ and $\Eing$ are, respectively, the
Ricci and Einstein tensors of $(\M,g)$,
$D$ is the Levi-Civita covariant derivative of $(S_{\mu},\gamSigmu)$,
$\Scalgam$ the corresponding curvature scalar, $\vec{H}$ is the mean
curvature vector 
\begin{equation*}
\vec{H} = - \frac{1}{\varphi} \left ( \thk \ellvarphi
+ \theta_{\ellvarphi} k \right ),
\end{equation*}
and the connection
of the normal bundle $\s_{\ellvarphi}$ is the one-form on $S_{\mu}$ defined by 
\begin{equation*}
\s_{\ellvarphi}(X) = \frac{1}{\varphi} \la   \nabla_{X} k ,\ellvarphi \ra, \quad 
\quad X \in \mathfrak{X}(S_{\mu}).
\end{equation*}
The evolution of $M(S_{\mu},\ellvarphi)$ in this general setting is given 
in the following lemma.
\begin{Lem}
\label{monot}
Let $\N$ be a null hypersurface embedded in a spacetime $(\M^4,\gMm)$.
Assume that $\N$ has topology $\Sig \times \mathbb{R}$
with the null generator tangent to the $\mathbb{R}$ factor. Consider a
foliation $\{ \Sig_{\mu}\}$ of $\N$ by spacelike hypersurfaces, all
diffeomorphic to $\Sig$. Let $k$ be the future null generator
satisfying $k(\mu) =-1$ and $\ellvarphi$ the null normal
to $\Sig_{\mu}$ satisfying $\la k,\ellvarphi \ra = -\varphi$. Then
\begin{align}
\frac{ d M(S_\mu,\ellvarphi)}{d \mu}
= \frac{1}{\sqrt{64\pi|S_\mu|}}
\int_{S_\mu} (-\thk) \bm{\eta_{S_\mu}}
+  \frac{1}{16 \pi} 
\int_{S_\mu} & \left [ \Eing(\ell,k) - \frac{\varphi}{2} \Scalgam
+ \varphi \left ( - \mbox{div}_{S_\mu} \s_{\ellvarphi} + |\s_{\ellvarphi}|^2_{\gamma_{S_\mu}}
\right )  \right . \nonumber \\
& \left . + \left ( \frac{1}{\varphi} k(\varphi) - Q_{k} \right ) \theta_{\ellvarphi}
\right ] 
\bm{\eta_{S_\mu}} \label{Var1}
\end{align}
where $Q_k$ is defined by $\nabla_k k = Q_k k$. If, moreover, $\varphi$ is constant and $k$ is 
geodesic ($Q_k=0$) then
\begin{align}
\frac{ d M(S_\mu,\ellvarphi)}{d \mu}
=  \frac{1}{\sqrt{64\pi|S_{\mu}|}}
\int_{S_{\mu}} (-\thk) \bm{\eta_{S_\mu}}
- \frac{\varphi \, \chi(\Sig)}{8}
+  \frac{1}{16 \pi} 
\int_{S_\mu} & \left ( \Eing(\ellvarphi,k) + \varphi |\s_{\ellvarphi}|^2_{\gamma_{S_{\mu}}}\right )
\bm{\eta_{S_\mu}}
\label{Var2}
\end{align}
where $\chi(\Sig)$ is the Euler characteristic of $\Sig$.
\end{Lem}
 
\begin{proof}
We drop all reference to $\mu$ for simplicity. The volume
form satisfies
\begin{align*}
k ( \volSig ) = \thk \volSig
\end{align*}
so that the variation along $-k$ of $M(S,\ellvarphi)$ is, using (\ref{kthl2}),
\begin{align*}
(-k) (M(S,\ellvarphi)) =  \frac{1}{\sqrt{64\pi|\Sig|}}
\int_S (-\thk) \volSig
+  \frac{1}{16 \pi} 
\int_{S} & \left [ \Eing(\ellvarphi,k) - \frac{\varphi}{2} \Scalgam
+ \varphi \left ( - \mbox{div}_{\Sig} \s_{\ellvarphi} + |\s_{\ellvarphi}|^2_{\gamSig}
\right )  \right . \\
& \left . + \left ( \frac{1}{\varphi} k(\varphi) - Q_{k} \right ) \theta_{\ellvarphi} 
\right ]
\volSig
\end{align*}
where we have used $\la \vec{H}, \vec{H} \ra = - \frac{2}{\varphi} 
\thk \theta_{\ellvarphi}$.
This is precisely (\ref{Var1}). When $\varphi = \mbox{const}$ and $Q_k=0$, 
(\ref{Var2}) follows directly from (\ref{Var1}) as a consequence of the 
Gauss-Bonnet theorem
$\int_{\Sig} \Scalgam \volSig = 4 \pi \chi(\Sig)$.
\end{proof}

Our purpose in deriving the general variation formula (\ref{Var1}) is
to show that 
indeed $\varphi = \mbox{const}$ and $Q_k =0$ are the only clear situations
leading to a (nearly) monotonic behaviour. Indeed, the divergence 
term $\mbox{div}_{\Sig} \s_{\ellvarphi}$ has no sign which strongly
suggests the choice 
$\varphi = \mbox{const}$. 
The term in $\theta_{\ellvarphi}$, which again has no 
sign a priori, suggest 
making the choice $Q_{k} =0$  (the seemingly more general condition
of making $\varphi$ constant only within the leaves and 
 $Q_k = \varphi^{-1} k(\varphi)$ is simply a reparametrization
of the previous one).

Under the dominant energy condition (DEC) on $(\mathcal{M},g)$ 
(namely, $-\Eing(X,\cdot)$ future causal $\forall X$ future causal),
this lemma implies that if $S$ is connected and non-spherical, then
$M(S_{\mu},\ellvarphi))$ is monotonically increasing along any geodesic flow
for any past expanding (i.e.
with $\theta_{-k} \geq 0$) null hypersurface. We will reserve the
symbol $\lambda$ for foliations $\{ S_{\lambda} \}$ associated
to geodesic generators $k$. 

For the Penrose inequality in an asymptotically
flat context, the spherical topology is the relevant one.
In this setting, $\mmlam$ is not always monotonic.
However, under certain circumstances
one can relate its value on the initial surface and its asymptotic
value at infinity. In fact, obtaining such relations will be the main theme
of this paper. We first need to specify our asymptotic 
conditions. We adopt here the same definitions as in \cite{MarsSoria2015}, where a detailed analysis of the
limit  of the Hawking energy along null flows was obtained. For the sake
of completeness, we briefly repeat the main definitions. The reader is referred
to \cite{MarsSoria2015} for further details.

We  make the global assumption that $\Omega = \mathbb{S}^2 \times \mathbb{R}$
with the geodesic null generator $k$ tangent along the 
$\mathbb{R}$-factor. Implicit in this condition is that, fixed
a cross section $S_0$ of $\Omega$ (necessarily of spherical topology),
the integral curve of $k$
starting at $p \in S_0$
has maximal domain $(-\infty,\lambda_{+}(p))$, i.e. the null
generators are past complete. 
After possibly removing portions of $\Omega$ lying to the future of $S_0$
we can assume
that $\Omega$ is foliated by the  level sets $\{ S_\lambda \}$
of the function $\lambda \in \mathcal{F}(\Omega)$ defined
by $k(\lambda) = -1$, $\lambda|_{S_0} =0$. Obviously, all such $S_{\lambda}$ 
are of spherical topology. 
The function $\lambda$ is called {\bf level set function of $k$}.
A null hypersurface $\Omega$ satisfying these properties is called
{\bf extending to past null infinity}.
The definition of asymptotic flatness involves so-called
Lie constant transversal tensors. A {\bf transversal} tensor
is a covariant tensor on $\Omega$ completely orthogonal to $k$. 
Any such tensor $T$ is in one-to-one correspondence with a family $T(\lambda)$
of covariant tensors on $S_{\lambda}$.  A transversal
tensor is {\bf positive definite} if each $T(\lambda)$ has this
property. 
A transversal tensor $T$ is {\bf Lie constant} iff $\pounds_{k} T =0$. 
Concerning decay at infinity, the symbol $T = o_n(\lambda^{-q})$ means
$\lambda^{i+q} (\pounds_{k})^i T = o(1)$ ( $i=0,1,\cdots,n$), 
i.e. it has a vanishing limit
as $\lambda \rightarrow \infty$ 
in a Lie-propagated basis $\{X_A\}$ of $T \Siglam$. $T=o^X_n(\lambda^{-q})$
means $\lambda^q \pounds_{X_{A_1}} \cdots \pounds_{X_{A_i}} T = o (1)$ for all 
$i$-tuple of indices $\{A_1, \cdots A_i\}$ and $i=0,1,\cdots,n$. Similar definitions hold for 
$O_n(\lambda^{-q})$ and $O^X_n(\lambda^{-q})$. 

The definition of asymptotic flatness we use imposes conditions only along
$\Omega$ and reads (cf. Definition 1 and Proposition 3 in
\cite{MarsSoria2015})
\begin{Def}
\label{AF}
A null hypersurface $\Omega$ in a spacetime $(\mathcal{M}^4,g)$ is
{\bf past asymptotically flat} if it 
extends to past null infinity and
there exists a choice of cross section
$S_0$ and null geodesic
generator $k$ with corresponding level set function $\lambda$ satisfying:
\begin{itemize}
\item[(i)]  There exist two symmetric 2-covariant transversal
and Lie constant  tensor fields $\qh$ (positive definite)
and $h$ 
such that
$\gammatilde := \gamma - \lambda^2 \qh - \lambda h$
is $\gammatilde = o_1(\lambda) \cap o^{X}_{2}(\lambda)$ 
\item[(ii)] There exists a transversal, Lie constant
one-form $\sone$ such that
$\stilde_{\ell} := \s_{\ell} - \frac{\sone}{\lambda} $
is $\stilde_{\ell} = o_1(\lambda^{-1})$.
\item [(iii)] Denote the Gauss curvature of $\qh$ by
$\K_{\qh}$. There exists a Lie constant function  
$\thlone$ such that
$\thltilde := \theta_{\ell} - \frac{2 \K_{\qh}}{\lambda} - \frac{\A}{\lambda^2}$ 
is $\thltilde = o(\lambda^{-2})$.
\item[(iv)]
The function $\Riemg(X_A,X_B,X_C,X_D)$
along $\Omega$ is such that 
$\lim_{\lambda\rightarrow \infty} \frac{1}{\lambda^2}\Riemg(X_A,X_B,X_C,X_D)$ exists while
its double trace satisfies
$2\Eing(k,\ell)- \, \Scal^g - \frac{1}{2} \Riemg(\ell,k,\ell,k) = 
o (\lambda^{-2})$. 
\end{itemize}
\end{Def}


We can now analyze the limit of $\mmvarphi$ at infinity. From item (i) in Definition \ref{AF}, it follows
that the volume form $\volSiglam$ of each $S_\lambda$ satisfies
\begin{equation}
\label{volsrlamform}
\volSiglam=\left( \lambda^2+\c\lambda+o(\lambda)  \right)\volsh,
\end{equation}
where the Lie constant function $\c$ is defined by the expansion
\begin{equation}
\label{thetakdev}
\theta_k=\frac{-2}{\lambda}+\frac{\c}{\lambda^2}+o(\lambda^{-2})
\end{equation}
which is a consequence of Definition \ref{AF} (see \cite{MarsSoria2015}).
The expressions become simpler if we introduce the area radius at infinity
as
\begin{equation*}
\rq^2:=\frac{1}{4\pi}\int_{\hat{S}}\volsh,
\end{equation*}
where $\hat{S}$ represents the surface $S$ endowed with the metric $\qh$. The area $\Siglam$ has the following expansion
\begin{equation*}
|\Siglam|=\int_{\Siglam}\volSiglam=\int_{\Sh} \left( \lambda^2+\c\lambda+o(\lambda)  \right)\volsh=4\pi\rq^2\lambda^2+\left( \int_{\hat{S}}\c\volsh  \right) \lambda+o(\lambda)
\end{equation*}
and therefore
\begin{equation}
\sqrt{|\Siglam|}=\sqrt{4\pi\rq^2}\,\lambda+\frac{\int_{\hat{S}}\c\volsh
}{2\sqrt{4\pi\rq^2}}+o(1). \label{firsttermmm}
\end{equation}
We next compute the asymptotic behaviour of the second term in $\mmvarphi$.
Using 
item (iii) in Definition \ref{AF} and noticing that
$\thlvarphi = \frac{\varphi}{2} \thl$ (because of the  scaling relation 
$\ellvarphi = \frac{\varphi}{2} \ell$), it follows (we are assuming $\varphi$
constant here and in what follows)
\begin{align}
\int_{\Siglam}\thlvarphi(\lambda) \volSiglam &=
\int_{\hat{S}}\left(\frac{\varphi \K_{\qh}}\lambda+\frac{\varphi \a}{2 \lambda^2}+o(\lambda^2)   \right)\left( \lambda^2+\c\lambda+o(\lambda)  \right)\volsh
\nonumber \\
 & =
4\pi \varphi \lambda+
\int_{\hat{S}} \left( \varphi \K_{\qh}\c+ 
\frac{\varphi}{2}\a      \right) \volsh + o(1).
\label{secterm} 
\end{align}
Combining (\ref{firsttermmm}) and (\ref{secterm})  into (\ref{functional})
gives
\begin{equation*}
\mmvarphi= \left ( \frac{\rq}{2} - \frac{\varphi}{4} 
\right ) \lambda+ 
\frac{1}{16 \pi}
\int_{\Sh}\left(\c\left(\frac{1}{\rq}
- \varphi \K_{\qh}
\right)- \frac{\varphi}{2} \a \right)\volsh +o(1).
\end{equation*}
This expression has a finite limit at 
infinity if and  only if the scaling of $\ellvarphi$ is chosen so that
$\varphi=2\rq$. This vector will be denoted by $\ell^{\star}$ and its relation
to the canonical $\ell$ is $\ell^{\star}=\rq\ell$. With this choice,
\begin{equation}
\label{dlimmm}
\underset{\lambda \to \infty}{\lim} \mmlamp =
\frac{1}{16 \pi}
\int_{\Sh}\left(\c\left(\frac{1}{\rq}
- 2 \rq \K_{\qh}
\right)- \rq \a \right)\volsh.
\end{equation}
It is useful to relate this limit 
to the corresponding limit of the Hawking energy along $\{S_\lambda\}$.
For any closed spacelike surface $S$ with spherical topology embedded
in a four dimensional spacetime, the Hawking energy is defined by 
\begin{equation*} 
m_H(S)=\sqrt{\frac{|S|}{16\pi}} \left (
1-\frac{1}{16\pi}\int_{S}\la \vec{H}, \vec{H} \ra \bm{\eta_{S}} \right ).
\end{equation*}
The limit of $m_H(S_{\lambda})$ was investigated in detail in
\cite{MarsSoria2015}. In particular, Theorem 5 in \cite{MarsSoria2015} gives
\begin{equation}
\label{hawklimm}
\underset{\lambda \to \infty}{\lim} m_H(S_\lambda)=\frac{-\rq}{16\pi} \int_{\hat{S}} \left(\K_{\qh}\c+\a \right)\volsh.
\end{equation}
Combining (\ref{dlimmm}) and (\ref{hawklimm}), the following
Proposition is proved:
\begin{Prop}
\label{limits}
With the choice $\ell^{\star} = \rq \ell$, the limits of $\mmlamp$ and $m_H(S_\lambda)$ are related by
\begin{equation}
\label{limmixt}
\underset{\lambda \to \infty}{\lim} \mmlamp=
\underset{\lambda \to \infty}{\lim} m_H(S_\lambda)+\frac{1}{16\pi}\int_{\Sh}\c\left( \frac{1}{\rq}-\rq \K_{\qh}    \right) \volsh.
\end{equation}
\end{Prop}

\begin{remark}
\label{coincidinglimits} 

There are two  interesting cases where
the limit $M(\Siglam,\ell^{\star})$ agrees with
the limit of the Hawking energy along the foliation. The first one
occurs when $\qh$ has positive constant curvature, in which case
the area radius $\rq$ and the Gauss curvature are related by
${\cal K}_{\qh} = \frac{1}{\rq^2}$ and the integrand in 
the second term of (\ref{limmixt}) vanishes. Such foliations
are called ``approaching to large spheres'' because
the geometry of the leaves tends, after a suitable rescaling, to the round spherical metric.
This situation is particularly
relevant because then the limit of the Hawking energy is the Bondi
energy measured by the observer defined by the foliation $\{S_{\lambda} \}$
(see \cite{Bartnik2004,PenroseRindler,Sauter2008,MarsSoria2015} for details).

The other case corresponds to those foliations
satisfying $\c=constant$. In this case we have
\begin{equation} \int_{\Sh}\c\left( \frac{1}{\rq}-\rq \K_{\qh}
\right) \volsh=\c\int_{\Sh}\left( \frac{1}{\rq}-\rq \K_{\qh}
\right) \volsh=\c(4\pi\rq-4\pi\rq)=0,
\end{equation} where in the second equality we have used the
Gauss-Bonnet theorem. We devote the next section to study in detail
geodesic foliations with constant $\c$.
\end{remark}

\section{GAB foliations and a Penrose type inequality}

As discussed in the introduction, 
Ludvigsen and Vickers
\cite{LudvigsenVickers1983} and Bergqvist
\cite{Bergqvist1997} considered the Penrose inequality for
null hypersurfaces. A fundamental ingredient of their
work involved geodesic foliations for which $\c$ vanishes identically.
As we will discuss below, such foliations are closely
related to geodesic foliations with
$\c$ constant. We devote this section to study
such foliations. Our main result is a 
Penrose-type inequality valid in full generality and which reduces
to the Penrose inequality when the foliation approaches large
spheres. Besides its intrinsic interest, the general
Penrose-type inequality helps 
also putting
the result of Ludvigsen \& Vickers and Bergqvist into a broader
perspective and clarifies both its scope and its range of
validity.

We first need a lemma showing that, no matter which geodesic
foliation is taken, the leading term $\c$ is always strictly positive.
This may seem to contradict (\ref{LVBexpansion}), but this is not the
case because $\lambda=0$ corresponds to a cross section on $\Omega$, while
the corresponding  condition for $r$ was not assumed
(and in fact does not hold) in (\ref{LVBexpansion}).

\begin{Lem}
Let $\Omega$ be a past asymptotically flat null hypersurface with a choice of affinely
parame\-tri\-zed null generator $k$
and corresponding level set function $\lambda$. Assume that the spacetime satisfies the dominant
energy condition, then $\c>0$.
\end{Lem}

\begin{proof}
Let $\{S_{\lambda}\}$ the geodesic foliation defined by $\lambda$
and consider the function $\rho(\lambda)=\theta_k |_{S_{\lambda}}
\lambda^2+2\lambda$. Using the Raychaudhuri
equation (\ref{Ray}), which can be written as
$\frac{d\theta_k(\lambda)}{d\lambda}=\frac{\theta_k^2}{2}+W$ with $W\geq0$ under DEC, the derivative of $\rho$
satisfies
\begin{equation*}
\rho'(\lambda)=2\lambda\theta_k+\lambda^2\left(\frac{1}{2}\theta_k^2+W \right)+2 \geq \frac{(\lambda\theta_k+2)^2}{2}\geq 0.
\end{equation*}
Since $\rho$ vanishes at $\lambda=0$, it follows that $0\leq \rho(\lambda) \leq \underset{\lambda \to \infty}{\lim}\rho(\lambda)=\c$
where the last equality follows from the expansion (\ref{thetakdev}). 
To show the strict inequality $\c>0$ we argue by contradiction.
Assume that there is some null geodesic
$\gamma_p$ in  $\Omega$ where $\c=0$. Then $\rho(\lambda)$ necessarily vanishes on this curve
and 
\begin{equation*}
\theta_k |_{\gamma_p(\lambda)}\lambda^2+2\lambda 
 =0 \quad \Longrightarrow \quad 
\theta_k |_{\gamma_p(\lambda)} =\frac{-2}{\lambda} \quad \Longrightarrow \quad
\underset{\lambda \to 0^{+}}{\lim}\theta_k |_{\gamma_p(\lambda)}=-\infty , 
\end{equation*} 
which is a contradiction to the smoothness of $\Omega$ at $S_0$.
\end{proof}

The following result deals with the existence of foliations 
with constant $\c$.
\begin{Lem}
\label{exis}
Let $\Omega$ be a past asymptotically flat null hypersurface with a choice of affinely
parame\-tri\-zed null generator $k$
and corresponding level set function $\lambda$. There exists a Lie constant positive
function $f\in\mathcal{F}(\Omega)$ and a reparametrization $\lambda=\f\lambda'$ such that the corresponding asymptotic term $\c'$ of $\theta_{k'}$ is constant.
\end{Lem}

\begin{proof}
It is well-known that
the null expansion $\theta_k|_p$ is a property of $\Omega$ at $p \in \Omega$, independent
of the cross section passing through $p$. We can thus transform
the expansion (\ref{thetakdev})  under the change of foliation $\lambda=\f\lambda'$ simply 
as
\begin{equation*}
\theta_k=\frac{-2}{\f}\frac{1}{\lambda'}+\frac{\c}{\f^2}\frac{1}{\lambda'^2}+o(\lambda'^{-2}).
\end{equation*}
The null generator associated to $\lambda'$ is $k'=\f k$ (because $k'(\lambda') = -1$) so that
\begin{equation}
\theta_{k'}=\frac{-2}{\lambda'}+\frac{\c}{\f}\frac{1}{\lambda'^2}+o(\lambda'^{-2}).
\label{thkprime}
\end{equation}
Since $\c>0$,  we can choose $\f=\frac{\c}{\ct}$ for any given constant $\ct >0$. The foliation $\{S_{\lambda'}\}$ has
$\c' =c$, as claimed.
\end{proof}

Note that, by construction, the foliation $\{ S_{\lambda'} \}$ in this lemma is also a geodesic foliation. Once $\c$ is constant, it can be
made zero by a constant shift of $\lambda$. Indeed, let $\lambda$ be an affine parameter and define $\lambda = \lambda' + \lambda_0$
with $\lambda_0$ constant. The null generator $k$ now remains unchanged and
\begin{equation*}
\theta_k=\frac{-2}{\lambda}+\frac{\c}{\lambda^2}+o(\lambda^{-2})
= \frac{-2}{\lambda'} + \frac{\c + 2 \lambda_0}{\lambda'^2} + +o(\lambda'^{-2}).
\end{equation*}
Thus, the coefficient $\c$ along a geodesic foliation can be made zero by a change of origin if and only if it is constant. As mentioned above,
Ludvigsen \& Vickers and Bergqvist considered foliations with vanishing $\c$. Such foliations arise naturally
in the context of conformal compactifications of null infinity and are related to the Bondi coordinates near null infinity. This motivates
the following definition.
\begin{Def}[\bf Geodesic Asymptotically Bondi Foliation associated to $S_0$]
Given a past asymptotically flat null hypersurface $\Omega$ with a choice of cross section $S_0$.
A geodesic foliation $\{S_\lambda\}$ is called  Geodesic Asymptotically Bondi (GAB) and associated to $S_0$ iff
\begin{itemize}
\item[(i)]  $S_{\lambda=0}=S_0$
\item[(ii)]  $\c$ is constant.
\end{itemize}
\end{Def}
In the following lemma we show that two GAB foliations associated to $S_0$ are necessarily related
by a constant rescaling of parameter, $\lambda = a \lambda'$ with $a \in \mathbb{R}^+$. Thus, the collection of surfaces
$\{ S_{\lambda} \}$ remain unchanged, and GAB foliations associated to a given $S_0$ are geometrically unique. Obviously,
when $S_0$ changes, the corresponding unique GAB foliation (which exists by Lemma \ref{exis}) also changes. 
\begin{Lem}[\bf  Uniqueness of GABs]
Let $\Omega$ be a past asymptotically flat null hypersurface and $S_0$ a cross section. Two GAB foliations
$\{ S_{\lambda} \}$ and $\{ S_{\lambda'} \}$  associated to $S_0$ are related by $\lambda = a \lambda'$
for some positive constant $a$.
\end{Lem}

\begin{proof}
Let $k$ and $k'$ be the null generators of $\{ S_{\lambda} \}$ and $\{ S_{\lambda'} \}$. Since both are geodesic, there exists
a Lie constant positive function $\f$ such that $k' = \f k$. We have shown in (\ref{thkprime}) that
\begin{equation*}
\theta_{k'}=\frac{-2}{\lambda'}+\frac{\c}{\f}\frac{1}{\lambda'^2}+o(\lambda'^{-2})
= \frac{-2}{\lambda'}+\frac{\c'}{\lambda'^2}+o(\lambda'^{-2}).
\end{equation*}
By definition of GAB foliation, both $\c$ and $\c'$ are constant. Thus $\f$ is a positive constant (say $a$) 
and the affine parameters are related by $\lambda = \f \lambda' = a \lambda'$. 
\end{proof}

The main result in the work by Ludvigsen and Vickers and Bergqvist can be formulated in terms of GABs as follows.
\begin{Tma}[\bf Ludvigsen \& Vickers \cite{LudvigsenVickers1983}, Bergqvist
\cite{Bergqvist1997}]
\label{LV}
Let $\Omega$ be a past asymptotically flat null hypersurface $\Omega$ in a spacetime
satisfying the DEC. Assume that $\Omega$ admits a weakly outer trapped cross section $S_0$.
If $\Omega$ admits a GAB foliation $\{S_{\lambda}\}$ 
associated to $S_0$ and  approaching large spheres,
then the  Penrose inequality 
\begin{equation*}
E_B \geq \sqrt{\frac{|S_0|}{16\pi}}
\end{equation*}
holds, where $E_B$ is the Bondi energy associated to the observer at infinity 
defined by the foliation $\{\Siglam\}$.
\end{Tma}

As mentioned in
the Introduction, the possibility that the foliation can be chosen to approach large spheres 
was been assumed implicitly in the work by Ludvigsen and Vickers. The necessity 
to add this restriction explicitly 
was noticed by Bergqvist. Since GAB foliations associated to a given $S_0$ are unique, the condition
of approaching large spheres is indeed a strong additional assumption, that will only be satisfied
in very special circumstances. It makes sense to study GAB foliations in detail dropping
the assumption of approaching large spheres. By doing this we will be able to obtain
an interesting Penrose-type inequality relating the area of $S_0$, not to the Bondi energy, but to the limit
of the Hawking energy along the foliation. Since the Hawking energy approaches the Bondi energy for asymptotically
spherical foliations, our result will automatically include Theorem \ref{LV} as a Corollary. In particular, this will
help to clarify
the role played by the asymptotically spherical condition in Theorem \ref{LV}.

We have shown in Proposition \ref{limits} (cf. Remark \ref{coincidinglimits})  that for GAB foliations,
the limit of the functional $M(S,\ell^{\star})$ is the same as the limit of the Hawking energy
at infinity. To obtain a Penrose-type inequality we need to relate the value of
$M(S,\ell^{\star})$ at the initial surface with its asymptotic value. The functional
$M(S_{\lambda}, \ell^{\star})$ is not monotonic, so this cannot be done straight away. However,
the computations in Lemma \ref{monot} suggest splitting $M(S,\lambda)$ in two terms, one of which
will be automatically monotonic. This is useful because we can then concentrate in studying the non-monotonic term.
Define
\begin{align*}
\D(S,\ellvarphi) \defi & \sqrt{\frac{|S|}{16 \pi}} - \frac{\varphi}{4}
\lambda, \\
M_b(S,\ellvarphi) \defi & \frac{\varphi}{4}
\lambda- \frac{1}{16 \pi} \int_{\Sig} \theta_{\ellvarphi} \volSig,
\end{align*}
so that $M(S,\ellvarphi) = \D(S,\ellvarphi) + M_b(S,\ellvarphi)$.  The computation in Lemma \ref{monot} implies that
for geodesic flows and $\varphi = \mbox{const}$ (recall that the cross sections of $\Omega$ are
topological spheres, so that $\chi(S)=2$)
\begin{align*}
\frac{d \D(\Siglam,\ellvarphi)}{d \lambda} & = 
\frac{1}{\sqrt{64\pi|\Siglam|}}
\int_{\Siglam} ( - \thk) \volSiglam - \frac{\varphi}{4}, 
\\
\frac{d M_b(\Siglam,\ellvarphi)}{d \lambda} 
& =  \frac{1}{16 \pi} 
\int_{\Siglam} \left ( \Eing(\ellvarphi,k) + \varphi |\s_{\ellvarphi}|^2_{\gamSiglam}\right )
\volSiglam \quad \quad (\geq 0 \quad \mbox{under DEC} ).
\end{align*}
A direct consequence of $M_b(S,\ellvarphi)$ being monotonically increasing is that its initial value 
is bounded above by its value at infinity. From (\ref{secterm}), this limit is given by
\begin{equation*}
\underset{\lambda \to \infty}{\lim} \M_b(S_{\lambda},\ellvarphi) =
- \frac{1}{16 \pi}
\int_{\hat{S}} \varphi \left( \K_{\qh}\c+ 
\frac{1}{2}\a      \right) \volsh,
\end{equation*}
which is finite irrespectively of the choice of $\ellvarphi$. On the other hand, $\D(\Siglam, \ellvarphi)$ is not
necessarily monotonic and its limit at infinity is finite only for the choice $\ell^{\star} = \rq \ell$ and given by
(see (\ref{firsttermmm}))
\begin{equation}
\underset{\lambda \to \infty}{\lim} \D(\Siglam,\ell^{\star}) = 
\frac{ \int_{\hat{S}}\c\volsh   }{16\pi\rq}.
\label{Dsaympdevelop}
\end{equation}
To bound $M(\Siglam,\ell^{\star})$ from above we need to find an upper bound for
$\D(\Siglam, \ell^{\star})$. In fact, we shall prove $\D(\Siglam,\ell^{\star}) \leq 
\underset{\lambda \to \infty}{\lim} \D(\Siglam,\ell^{\star})$ provided the foliation $\{\Siglam\}$ is
GAB. In the following lemma we introduce a functional that turns out to be monotonic
for GAB foliations.
\begin{Lem}
\label{limFlemma}
Let $\Omega$ be a past asymptotically flat null hypersurface with a choice of affinely
parame\-tri\-zed null generator $k$
and corresponding level set function $\lambda$. Assume that the spacetime satisfies the dominant energy condition. Consider the functional 
\begin{equation*}
F(\Siglam)=\frac{|S_\lambda|}{\left(8\pi\rq^2\lambda+\int_{\hat{S}}\c\volsh\right)^2}.
\end{equation*}
If $\{S_\lambda\}$ is the GAB foliation associated to $S_0$, then $F(\Siglam)$ is monotonically increasing.
\end{Lem} 
\begin{proof}

Writing $F(\Siglam)$ as
\begin{equation*}
F(\Siglam) =\int_{\Siglam}\frac{\volSiglam}{\left(8\pi\rq^2\lambda+\int_{\hat{S}}\c\volsh  \right)^2}
\end{equation*}
and using $\frac{d}{d\lambda}\volSiglam=-\theta_k \volSiglam$, the derivative of $F(\Siglam)$ is
\begin{eqnarray}
\frac{d}{d\lambda}F(\Siglam)&=&\int_{\Siglam}\left (\frac{-16\pi\rq^2}{\left(8\pi\rq^2\lambda+\int_{\hat{S}}\c\volsh  \right)^3}+\frac{-\theta_k}{\left(8\pi\rq^2\lambda+\int_{\hat{S}}\c\volsh  \right)^2}           \right)\volSiglam=  \nonumber \\
&&
\int_{\Siglam}\left( \frac{-16\pi\rq^2+(-\theta_k)\left(8\pi\rq^2\lambda+\int_{\hat{S}}\c\volsh  \right)}{\left(8\pi\rq^2\lambda+\int_{\hat{S}}\c\volsh  \right)^3} \right)\volSiglam. 
\label{FderivwithR}
\end{eqnarray}
This derivative is non-negative provided
\begin{equation}
\label{sufuno}
(-\theta_k)\left(8\pi\rq^2\lambda+\int_{\hat{S}}\c\volsh  \right)\geq 16\pi\rq^2 \quad \Longleftrightarrow 
\quad \frac{1}{-\theta_k}-\frac{\lambda}{2}\leq \frac{1}{16\pi\rq^2}\int_{\Sh}\c \volsh. 
\end{equation}
The Raichaudhuri equation (\ref{Ray}) implies that the function 
$ \frac{1}{-\theta_k}-\frac{\lambda}{2}$ has non-negative derivative (under DEC). Since its limit
at infinity is $\frac{\c}{4}$ it follows
\begin{equation}
\label{sufdos}
 \frac{1}{-\theta_k}-\frac{\lambda}{2}\leq \frac{\c}{4}.
\end{equation}
which holds true for any geodesic foliation. For GAB foliations we have, using $\int_{\Sh} \volsh = 4 \pi \rq^2$,
\begin{equation*}
\frac{\c}{4} = \frac{1}{16\pi \rq^2}\int_{\Sh}\c \volsh 
\end{equation*}
and (\ref{sufdos}) is exactly (\ref{sufuno}).
\end{proof}
 
The monotonicity of the functional $F(\Siglam)$ is useful to establish an upper bound for $D(\Siglam,\ell^{\star})$, irrespectively
of whether the foliation is GAB o not.
\begin{Lem}
\label{LemMonF}
Let $\{S_{\lambda} \}$ be a geodesic foliation with leading term  metric $\qh$. If the functional $F(\Siglam)$
is monotonically increasing, then
\begin{equation}
  D(S_{\lambda}, \ell^{\star}) \leq \underset{\lambda \to \infty}{\lim}D(S_{\lambda}, \ell^{\star}).
\label{RelDimD}
\end{equation}
\end{Lem}

\begin{proof}
The monotonicity of the functional $F(\Siglam)$ along $\{\Siglam\}$ implies
\begin{equation}
\label{ineqlim}
F(\Siglam)\leq \underset{\lambda \to \infty}{\lim} F(\Siglam).
\end{equation}    
To compute this limit we use
\begin{equation*}
|S_{\lambda}|=\int_{S_{\lambda}}\bm{\eta_{S_{\lambda}}}=\int_{\Sh}(\lambda^2+\c \lambda+o(\lambda)){\bm{\eta_{\hat{q}}}}= 4\pi R_{\qh}^2 
\lambda^2+o(\lambda)
\end{equation*}
which follows from (\ref{volsrlamform})  and  $\int_{\hat{S}}\bm{\eta_{\hat{q}}} = 4 \pi \rq^2$. Hence
\begin{equation*}
 \underset{\lambda \to \infty}{\lim}F(\lambda)=\underset{\lambda \to \infty}{\lim} 
\frac{|S_{\lambda}|}{\left(8\pi R_{\qh}^2 \lambda+
\int_{\Sh}\c{\bm{\eta_{\hat{q}}}} \right)^2}
=\frac{1}{16\pi R_{\qh}^2}
\end{equation*}
and (\ref{ineqlim}) yields
\begin{equation}
\label{GABareabound}
\frac{|S_{\lambda}|}{\left(8\pi R_{\qh}^2 \lambda+
\int_{\Sh}\c{\bm{\eta_{\hat{q}}}} \right)^2}
\leq \frac{1}{16\pi R_{\qh}^2}
\quad \Longleftrightarrow \quad
\frac{|S_{\lambda}|}{16 \pi} \leq \left (  \frac{\rq}{2} \lambda 
+ \frac{1}{16 \pi \rq} \int_{\Sh}\c{\bm{\eta_{\hat{q}}}} \right)^2.
\end{equation}   
From the definition of $D(\Siglam, \ell^{\star})$  and using $\varphi = 2 \rq$
 we have
\begin{equation*}
D(S_{\lambda},\ell^{\star})=\sqrt{\frac{|S_{\lambda}|}{16\pi}}-\frac{R_{\qh}}{2}\lambda 
\leq 
\frac{ \int_{\hat{S}}\c\volsh   }{16\pi\rq}
\end{equation*}
after using (\ref{GABareabound}). Since the right-hand side is the limit of $\D(\Siglam)$ at infinity (\ref{Dsaympdevelop}),
we conclude  (\ref{RelDimD}).
\end{proof}
We can now establish our main result concerning GAB foliations.


\begin{Tma}[\bf A Penrose type inequality for GAB foliations]
\label{tmaprincipal}
Let $\Omega$ be a past asymptotically flat null hypersurface 
and $S_0$ a cross section. Assume that the spacetime satisfies the dominant energy condition. Then, the area $|S_0|$ satisfies the bound
\begin{equation}
\label{finalineq}
 \sqrt{\frac{|S_0|}{16\pi}} -\frac{1}{16\pi} \int_{S_0} \theta_{\ell^\star}
\bm{\eta_{S_0}}
\leq \underset{\lambda \to
\infty}{\lim} m_H(S_{\lambda}),
\end{equation}
where the limit is taken along the GAB foliation $\{S_{\lambda}\}$
associated to $S_0$. In particular, if $S_0$ is
a weakly outer trapped cross section, then
\begin{equation}
\label{mainineqa}
\sqrt{\frac{|S_0|}{16\pi}} \leq \underset{\lambda \to \infty}{\lim} m_H(S_{\lambda}).
\end{equation} 
\end{Tma}

\begin{proof}
From Lemmas \ref{limFlemma} and \ref{LemMonF}, $D(\Siglam,\ell^{\star})$ is bounded above by its limit
at infinity. The monotonicity of $M_b(S_{\lambda},\ell^\star)$ then implies
\begin{equation*}
M(S_{\lambda},\ell^\star) \leq \underset{\lambda \to \infty}{\lim} M(S_{\lambda},\ell^\star)= \underset{\lambda \to \infty}{\lim} m_H(S_{\lambda}),
\end{equation*}
where the last equality follows from Proposition \ref{limits}, since 
$\{\Siglam\}$ is GAB.
In particular, for $\lambda=0$ we have (\ref{finalineq}).
For the last statement we simply use that 
$\theta_{\ell^{\star}} \leq 0$ for
weakly outer trapped surfaces.
%
\end{proof}

Inequality (\ref{mainineqa}) gives a completely general upper bound for the area
of weakly outer trapper surfaces $S_0$
in terms of an energy-type quantity evaluated at infinity along the outward past null hypersurface 
generated by $S_0$, provided the latter stays regular all the way to infinity. In combination
to the general analysis of the limit of the Hawking energy at infinity 
carried out in \cite{MarsSoria2015}, this provides a Penrose-type inequality with potentially
interesting consequences. Obviously, this inequality will
only correspond to the Penrose inequality whenever the limit of the Hawking energy
agrees with the Bondi energy of the cut at infinity defined by $\Omega$. 
As already mentioned,
this is known to occur 
\cite{Bartnik2004,PenroseRindler,Sauter2008,MarsSoria2015}
for foliations approaching large spheres.
When $\Omega$ admits
a GAB foliation approaching large spheres, then  the limit of the Hawking energy along
this foliation is the Bondi energy $E_B$ associated to this observer at infinity,
and the Penrose-type inequality in Theorem \ref{tmaprincipal}
becomes the standard Penrose inequality, thus
recovering the original result by Ludvigsen \& Vickers and Bergqvist quoted as
Theorem \ref{LV}.

\section{On the inequality 
$\bm{D(\Siglam,\ell) \leq 
\underset{\lambda \to \infty}{\lim} D(\Siglam, \ell)}$.
}

The key ingredient that allowed us to prove the Penrose-type
inequality (\ref{mainineqa}) is
$D(\Siglam,\ell^{\star}) \leq  \underset{\lambda \to 
\infty}{\lim} D(\Siglam, \ell^{\star})$. 
In fact, the argument in the proof of Theorem \ref{tmaprincipal} 
combined with Proposition \ref{limits}
shows that any surface $S_0$ satisfying the inequality
\begin{equation}
D(S_0, \ell^{\star}) \leq 
\underset{\lambda \to \infty}{\lim} D(\Siglam, \ell^{\star})
= \frac{1}{16 \pi \rq} \int_{\hat{S}} \c  \volsh 
\label{problem} 
\end{equation} 
along a geodesic foliation  $\{ \Siglam \}$ starting at $S_0$,  in a spacetime satisfying the dominant energy condition, 
automatically satisfies the inequality
\begin{equation}
\sqrt{\frac{|S_0|}{16 \pi}} - \frac{1}{16 \pi} 
\int_{S_0} \theta_{\ell^{\star}}
\bm{\eta_{S_0}} \leq 
\underset{\lambda \to \infty}{\lim} m_H(S_\lambda)+\frac{1}{16\pi}\int_{\Sh}
\c\left( \frac{1}{\rq}-\rq \K_{\qh}    \right) \volsh.
\label{mainineq}
\end{equation}
This is the Penrose inequality provided $S_0$
is a weakly outer trapped surface and 
the right hand is the Bondi energy $E_B$  along $\{ \Siglam\}$. 
For this it is sufficient that $\{ \Siglam\}$ approaches large spheres
and this will be the case we will be interested from now on.
However, we postpone making the assumption that 
$\qh$ is the round metric until subsection \ref{RAM} because Proposition
\ref{Prop2} below (which holds for arbitrary geodesic foliations) may
be of independent interest.

In the previous section the validity of (\ref{problem})
followed from the monotonicity of $F(\Siglam)$
along GAB foliations. As shown in Lemma
\ref{LemMonF} monotonicity
of $F(\Siglam)$ is sufficient to establish (\ref{problem})
for arbitrary geodesic foliations. Since the
derivative of 
(\ref{FderivwithR}) is
\begin{equation*}
\frac{d}{d\lambda}F(S_\lambda)=\frac{1}{\left( 8\pi\rq^2\lambda+\int_{\Sh}\c\volsh\right)^2} \left( \frac{d}{d\lambda}|S_\lambda|- \frac{16\pi \rq^2 |S_\lambda|}{8\pi\rq^2\lambda+\int_{\Sh}\c\volsh  }\right) 
\end{equation*}
we have established:
\begin{Prop}
\label{Prop2}
Let $\Omega$ be a past asymptotically flat null hypersurface
in a spacetime satisfying the dominant energy condition
and $\{ \Siglam \}$ a geodesic foliation. If  
\begin{equation}
\label{sufcondPenrose}
\frac{d}{d\lambda}|S_\lambda| \geq \frac{16\pi \rq^2 |S_\lambda|}{8\pi\rq^2\lambda+\int_{\Sh}\c\volsh  }
\end{equation}
holds for all $\lambda \geq 0$ then the
inequality (\ref{mainineq}) holds. In particular if $\{ \Siglam\}$
approaches large spheres and (\ref{sufcondPenrose}) is satisfied, then 
the Penrose inequality
$E_B \geq \sqrt{\frac{|S_0|}{16 \pi}}$ holds, where $E_B$
is the Bondi energy associated to the observer defined by $\{ \Siglam \}$.
\end{Prop}

\begin{remark}
\label{Prop2bis}
Expanding the area as 
\begin{align}
|S_\lambda|=4\pi\rq^2\lambda^2+\left( \int_{\Sh}\c\volsh \right)\lambda+\hat{\Theta},
\label{AreaExp}
\end{align}
(\ref{sufcondPenrose}) becomes, after
some cancellations,
\begin{equation}
\label{firstarticleineqgeneral}
\left(\int_{\Sh}\c\volsh\right)^2+\left( 8\pi\rq^2\lambda+\int_{\Sh}\c\volsh\right)\frac{d\hat{\Theta}}{d\lambda} \geq 16\pi\rq^2\hat{\Theta}.
\end{equation}
This alternative form of Proposition \ref{Prop2} will be used
in Section \ref{RAMMink} below.
%

\end{remark}

GAB foliations have the property that (\ref{sufcondPenrose}) is always true.
It is natural to ask whether the constancy of $\c$ can be relaxed and
still obtain sufficiently general
conditions under which 
(\ref{sufcondPenrose}) holds. The issue, however, appears to be difficult.
In the next subsection we study the behaviour of the derivative of $F(S_\lambda)$   
near infinity and show that both cases of $F(\Siglam)$ being
monotonically increasing or monotonically
decreasing near infinity are possible.

%

\subsection{On the monotonicity of
$F(\Siglam)$ for large $\lambda$}

A necessary condition for (\ref{FderivwithR}) to be non-negative for all 
$\lambda$ is, of course, that its leading term at infinity is non-negative.
To determine the asymptotic behaviour at infinity requires
one extra term in the expansion of $\theta_k$ as compared to (\ref{thetakdev}).
To make sure this is possible we need a slightly stronger definition
of asymptotic flatness.

%

\begin{Def}[\bf Strong past asymptotic flatness]
\label{SAF}
A null hypersurface $\Omega$ in a
spacetime $(\mathcal{M}^4,g)$ 
is  {\bf strong past asymptotically flat} if it is past asymptotically flat
with $(i)$ in Definition \ref{AF} replaced by the stronger condition
\begin{itemize}
\item[(i)']  There exist symmetric 2-covariant transversal
and Lie constant  tensor fields $\qh$ (positive definite), $h$ and $\Psi_0$ 
such that $\gammatilde$ defined by
$\gamma= \lambda^2 \qh+\lambda h+ \Psi_0 +\gammatilde$ is
$\gammatilde=o_1(1)\cap o_2^X(1)$.
\end{itemize}
\end{Def}

\begin{remark}
In strong asymptotically flat null hypersurfaces, 
all geodesic foliations $\{ \Siglam \}$ automatically
satisfy item (i)' in the definition. Also, there always exist geodesic
foliations  $\{\Siglam\}$ for which the asymptotic metric
$\qh$ is the round metric of unit radius on $\mathbb{S}^2$ (see
\cite{MarsSoria2015} for a proof of both facts in the context
of asymptotically
flat null hypersurfaces, which carries over
immediately to the strong  asymptotically flat case).

\end{remark}

A first consequence of strong asymptotic flatness is that 
the function $\Theta$ defined by
\begin{equation}
\label{volSiglongexp}
\volSiglam= (\lambda^2+\c\lambda+\Theta)\volsh
\end{equation}
is of the form $\Theta = \Theta_0 + \tilde{\Theta}$ with
$\Theta_0$ Lie constant and $\tilde{\Theta} = o_1(1)$. A second consequence,
which follows from (\ref{Evolmetric}), is that 
$\theta_k$ admits the expansion
\begin{equation}
\label{thetaklongexp}
\theta_k=\frac{-2}{\lambda}+\frac{\c}{\lambda^2}+\frac{\cdos}{\lambda^3}+o(\lambda^{-3})
\end{equation}
with $\cdos$ Lie constant. The following proposition relates $\Theta_0$ with 
$\c$ and $\cdos$ and provides a universal bound for $\cdos$.
\begin{Prop}
\label{lemaprimero}
Let $\Omega$ be a strong past asymptotically flat null hypersurface 
and $\{ \Siglam \}$ a geodesic foliation. Then
\begin{align}
\label{twocondit}
\Theta_0 & :=\underset{\lambda \to \infty}{\lim}\Theta=\frac{1}{2}\left(\left(\c \right)^2+\cdos \right)
\end{align}
If in addition the spacetime satisfies the dominant energy condition then we will also have
\begin{equation}
\label{coefdos}
\cdos\leq -\frac{1}{2}\left( \c\right)^2\leq 0.
\end{equation}
\end{Prop}

\begin{proof}
Inserting (\ref{volSiglongexp}) and (\ref{thetaklongexp})
into the evolution equation
\begin{equation}
\label{difvolform}
\frac{d}{d\lambda}\volSiglam=-\theta_k\volSiglam
\end{equation}
gives
%
\begin{align*}
(2\lambda+\c+\frac{d\tilde{\Theta}}{d\lambda})\volsh
& =-\left(\frac{-2}{\lambda}+\frac{\c}{\lambda^2}+\frac{\cdos}{\lambda^3}+o(\lambda^{-3}) \right)(\lambda^2+\c\lambda+\Theta_0+o(1))\volsh \\
& = \left ( 2\lambda+\c+\left(2\Theta_0-\left(\c\right)^2-\cdos \right)\frac{1}{\lambda}+o(\lambda^{-1}) \right ) \volsh.
\end{align*}
Since $\frac{d\tilde{\Theta}}{d\lambda}=o(\lambda^{-1})$ we conclude
$2\Theta_0-\left(\c\right)^2-\cdos=0$,  which proves (\ref{twocondit}).
For the universal bound (\ref{coefdos}), let us define
$\hat{I}(\lambda)=\frac{\volSiglam}{(\c+2\lambda)^2}$. Its derivative
is
\begin{equation*}
\frac{d}{d\lambda}\hat{I}(\lambda)=\frac{1}{(\c+2\lambda)^2}\left(-\theta_k-\frac{4}{\c+2\lambda}\right)\volSiglam \geq 0
\end{equation*}
where in the last inequality we used (\ref{sufdos}) (here is where
the DEC is used). $\hat{I}(\lambda)$
has limit at infinity $\frac{1}{4} \volsh$. In combination with 
the fact that $\hat{I}$ is monotonically increasing we conclude
\begin{equation*}
\hat{I}(\lambda)=\frac{\volSiglam}{(\c+2\lambda)^2}\leq 
\frac{1}{4} \volsh.
\end{equation*}
Inserting (\ref{volSiglongexp}), a direct computation gives
\begin{equation*}
\left(-\frac{1}{16} \left(\c\right)^2+ \frac{1}{4} \Theta_0 \right)\frac{1}{\lambda^2}+o(\lambda^{-2})\leq 0 \quad \quad \Longrightarrow \quad \quad 
\Theta_0\leq \frac{1}{4} \left(\c\right)^2,
\end{equation*}
which is simply (\ref{coefdos}) after using the explicit form of $\Theta_0$.
\end{proof}

Let us now find the asymptotic expansion of the right hand side of
(\ref{FderivwithR}). Plugging the
asymptotic expansion (\ref{thetaklongexp}) gives, after a straightforward computation,
\begin{equation}
\label{expansionderivF}
\frac{d}{d\lambda}F(\Siglam)=\frac{1}{(8 \pi\rq^2)^3}\left(  \left( \int_{\Sh}\c\volsh \right)^2-8\pi\rq^2\int_{\Sh}\left(\c\right)^2\volsh-8\pi\rq^2\int_{\Sh}\cdos\volsh  \right)\frac{1}{\lambda^3}+o(\lambda^{-3}).
\end{equation} 
The leading coefficient can be rewritten as
\begin{equation*}
\frac{F_{\infty}}{(8 \pi\rq^2)^3}  :=
\frac{1}{(8\pi\rq^2)^3}\left(
-4\pi\rq^2\int_{\Sh}\left( \left(\c \right)^2+2\cdos \right) \volsh +
\left [
  \left( \int_{\Sh}\c\volsh \right)^2-4\pi\rq^2\int_{\Sh}\left(\c\right)^2\volsh \right ] \right) 
\end{equation*}
which is a difference of positive quantities. Indeed, the first integral
is non-negative because of  (\ref{coefdos}), while the term
in brackets is non-positive because
\begin{equation*}
\Big ( \int_{\Sh}\c\volsh \Big )^2 \leq 4\pi\rq^2\int_{\Sh} \left (\c \right 
)^2\volsh 
\end{equation*}
by the H\"older inequality. Depending on which
term dominates, the functional $F(\Siglam)$ will be 
increasing or decreasing near infinity. Non-negativity of the
leading term (\ref{expansionderivF}) 
is obviously a necessary condition for
the hypothesis of Proposition \ref{Prop2} to hold. However, even
when (\ref{expansionderivF})  has the right sign, it is not at
all obvious how to ensure that $F(\Siglam)$ is monotonic
for all $\lambda$ when the foliation is, in addition, assumed to approach
large spheres. We have attempted (and failed) finding sufficient
condition ensuring 
$\frac{d^2}{d\lambda^2}F(\Siglam)\leq 0$,
as this would immediately imply that $F(S_\lambda)$ is increasing 
(because $F'(\Siglam) \rightarrow 0$ at infinity). Despite
the lack of success so far, approaching the null Penrose inequality
using the monotonic functional $F(\Siglam)$ remains an interesting
open problem, specially in view of the fact that $F(\Siglam)$
is {\it always} monotonic for GAB foliations.


\subsection{Renormalized Area Method for the Penrose inequality}

\label{RAM}

Monotonicity of $F(\Siglam)$  along geodesic foliations
approaching large spheres is an interesting sufficient
condition for the Penrose inequality along null hypersurfaces.
However, as discussed
in the previous subsection, it appears to be difficult 
to find general situations where $F'(\Siglam) \geq 0$ can be guaranteed.
In this subsection we consider an  a priori different  set up which 
implies the validity of (\ref{problem}) and hence
of the Penrose inequality whenever the
foliations also satisfies the restriction
of approaching large spheres. Let us assume from now on that $\qh$ is a
round metric on the sphere. Without loss of generality we can then
assume that $\qh$ is a round metric
of radius one, which we denote by $\q$.
Then $\rq=1$ and $\ell^{\star} = \ell$.
We want to investigate the condition
\begin{equation}
\label{Dcondition}
\frac{d}{d\lambda}D(S_\lambda,\ell^\star)\geq 0
\end{equation}
which indeed implies the validity of (\ref{problem})  and hence
the validity of the Penrose inequality.
Since $|\Siglam|$  diverges at infinity like $4\pi \lambda^2$, 
the functional $D(\Siglam,\ell)$ can be regarded as a renormalization of
the area functional, in order to make it bounded. We thus call
the approach to the null Penrose inequality via (\ref{Dcondition})
the {\bf renormalized area method}.  It is interesting that 
this method is, in fact, a subcase of the general setup 
involving monotonicity of $F(\Siglam)$.

\begin{Prop}
\label{dfimplication}
Let $\Omega$ be a strong past asymptotically flat null hypersurface
and $\{\Siglam\}$ a geodesic foliation approaching large spheres. Then
\begin{equation*}
\frac{d}{d\lambda}D(S_\lambda,\ell)\geq 0\quad \Longrightarrow 
\quad \frac{d}{d\lambda}F(S_\lambda)\geq 0.
\end{equation*}
\end{Prop}
\begin{proof}
Let $L:=\frac{1}{16\pi}\int_{\esfdos}\c\volunitdos>0$ be the limit
of $D(S_\lambda,\ell)$ at infinity. Since $|\Siglam| = 
16 \pi \left ( D(\Siglam,\ell)
+ \frac{\lambda}{2} \right )^2$ we can rewrite $F(\Siglam)$ as
\begin{equation*}
F(S_\lambda)=\frac{|S_\lambda|}{\left(8\pi\lambda +\int_{\esfdos} \c\volunitdos\right)^2}=\frac{(\Dlam+\frac{\lambda}{2})^2}{16\pi(L+\frac{\lambda}{2})^2}
\end{equation*} 
where $\Dlam$ is a short-hand for $D(\Siglam,\ell)$. Let
\begin{equation*}
f(\lambda):=\sqrt{16\pi F(S_\lambda)}=\frac{\Dlam+\frac{\lambda}{2}}{L+\frac{\lambda}{2}}
\end{equation*}
so that 
\begin{equation}
\label{DandFequation}
f'(\lambda)\left(L+\frac{\lambda}{2}\right) =
\frac{d \Dlam}{d\lambda} +\frac{1}{2}(1 - f(\lambda)).
\end{equation}
If  $\frac{d \Dlam}{d\lambda} \geq 0$
it follows $\Dlam \leq\underset{\lambda \to \infty}{\lim}\Dlam=L$ so that
$f(\lambda)=\frac{D+\frac{\lambda}{2}}{ L+\frac{\lambda}{2}}\leq 1$
and we conclude from (\ref{DandFequation}) that $f'(\lambda) \geq 0$, which is 
is equivalent to $F'(\Siglam) \geq 0$.
\end{proof}

The derivative of $D(\Siglam,\ell)$ is
\begin{equation*}
\frac{d}{d\lambda}D(S_\lambda)
=\frac{1}{2\sqrt{16\pi\Siglam}}\left(\frac{d}{d\lambda}|\Siglam|-\sqrt{16\pi\Siglam}\right)=\frac{1}{2\sqrt{16\pi\Siglam}}\left(\int_{\Siglam}(-\theta_k)\volSiglam-\sqrt{16\pi\Siglam}\right).
\end{equation*}
Given that $\theta_k <0$, the inequality 
$\frac{d}{d\lambda}D(S_\lambda) \geq 0$ can be equivalently
written in a slightly more convenient from
as $G(\lambda)\geq0$, where
\begin{equation*}
G(\lambda):=\left( \int_{\Siglam}(-\theta_k)\volSiglam \right)^2-16\pi|\Siglam|.
\end{equation*}
We start by computing the limit of $G(\lambda)$ at infinity.
\begin{Prop}
\label{Prop5}
With the same assumptions as in Proposition \ref{dfimplication},
\begin{equation}
\label{glimit}
 \underset{\lambda \to \infty}{\lim}G(\lambda)=
F_{\infty} = \left( \int_{\esfdos}\c\volunitdos \right)^2-8\pi\int_{\esfdos}\left(\c\right)^2\volunitdos-8\pi\int_{\esfdos}\cdos\volunitdos. 
\end{equation}
\end{Prop}

\begin{proof}
We have shown in Proposition \ref{lemaprimero} that 
\begin{equation}
\label{volsiglamlongexp}
\volSiglam= \left (\lambda^2+\c\lambda+
\frac{1}{2}\left(\left(\c{} \right)^2+\cdos \right ) 
+o(1) \right )\volunitdos.
\end{equation}
From expansion (\ref{thetakdev}), we have $\theta_k \volSiglam =
- 2 \lambda - \c + o(1)$, so that
\begin{equation*}
\left(\int_{\Siglam}\theta_k \volSiglam\right)^2=64\pi^2\lambda^2+16\pi
\lambda \int_{\esfdos}\c\volunitdos + \left(\int_{\esfdos}\c\volunitdos \right)^2+o(1).
\end{equation*}
Also form (\ref{volsiglamlongexp}),
\begin{equation*}
|\Siglam| = 4 \pi \lambda^2 + \lambda \int_{\esfdos} \c \volunitdos
+ \int_{\esfdos} \frac{1}{2}\left(\left(\c \right)^2+\cdos \right)
\volunitdos+ o(1).
\end{equation*}
Inserting both into $G(\lambda)$ the divergent terms cancel out and we are left
with (\ref{glimit}).
\end{proof}

\begin{remark}
\label{EquivInfty}
The limit of $G(\lambda)$ is directly related to
the leading term in the asymptotic expansion of $F(\Siglam)$
so that the inequality ``at infinity'' $F_{\infty} \geq 0$
is necessary for both methods. Thus, for sufficiently large $\lambda$,
the renormalized area method does not only imply $F'(\Siglam) \geq 0$,
but it is in fact equivalent to it (possibly excluding the case
$F_{\infty}=0$ where higher order terms dominate). However, we do not expect this to 
be true for all $\lambda$, as it appears that $D'(\Siglam,\ell)
\geq 0$ should be a proper subset of $F'(\Siglam)\geq 0$.
\end{remark}

Assuming we are in the situation where 
$\underset{\lambda \to \infty}{\lim}G(\lambda) \geq 0$,
we can ensure 
$G(\lambda) \geq 0$ by the condition $G'(\lambda) \leq 0$.
This derivative is, from the
Raychaudhuri equation (\ref{Ray}),
\begin{align*}
G'(\lambda) & = 2\left( \int_{\Siglam}\theta_k\volSiglam\right)\left( \frac{d}{d\lambda}\left( \int_{\Siglam}\theta_k\volSiglam\right)+8\pi   \right) \nonumber \\
& = 2\left( \int_{\Siglam}\theta_k\volSiglam\right) \left( \int_{\Siglam}\left(\Ricg(k,k)-\frac{1}{2}\theta_k^2+\Pi^k_{AB}{\Pi^k}^{AB}\right)\volSiglam+8\pi\right).
\end{align*}
Since the first term is always negative, $G'(\lambda) \leq 0$
is equivalent to $H(\lambda) \geq0 $, where we have defined
\begin{equation*}
H(\lambda):=\int_{\Siglam}\left(\Ricg(k,k)-\frac{1}{2}{\theta_k}^2+\Pi^k_{AB}{\Pi^k}^{AB}\right)\volSiglam+8\pi.
\end{equation*}
We proceed with the computation
of the derivative of this function and of its
limit at infinity.
\begin{Prop}
\label{Prop6}
With the same assumptions as in Proposition \ref{dfimplication},
$\underset{\lambda \to \infty}{\lim}H(\lambda)=0$ and
the derivative of $H(\lambda)$ is
\begin{equation}
\label{derivh}
H'(\lambda)=\int_{\Siglam}\left( -2\theta_k\Ricg(k,k)+2(\Pi^k)^{AB}R_{AB}+\frac{d}{d\lambda}\Ricg(k,k)\right)\volSiglam,
\end{equation}
where $R_{AB}:=\Riemg(X_A,k,X_B,k)$.
\end{Prop}

\begin{proof}
For the limit, we split $H(\lambda)$ in three terms 
and show that each one tends to zero.  
We start with $\int_{\Siglam} \Pi^k_{AB}{\Pi^k}^{AB} \volSiglam$. 
From equation (\ref{Evolmetric})
and the expansion (i) in Definition \ref{SAF} for the metric $\gamma$,
it follows
\begin{equation}
\label{Kkexpansion}
K^k_{AB}=-\q_{AB}\lambda-\frac{1}{2}h_{AB}+o(1)
\end{equation}
so that its trace-free part is
$\Pi^k_{AB}=O(1)$.
Since $\gamma(\lambda)_{AB} = \lambda^2 \q_{AB} + o(\lambda)$, its inverse 
is 
\begin{equation}
\label{invmetricexpansion}
\gamma(\lambda)^{AB}=\frac{1}{\lambda^2}\q^{AB}+o(\lambda^{-2})
\end{equation}
and $\Pi^k_{AB} \Pi^k{}^{AB} = O(\lambda^{-4})$ so that
\begin{equation*}
\int_{\Siglam}\left(\Pi^k_{AB}{\Pi^k}^{AB}  \right)\volSiglam
\stackrel{\lambda \rightarrow \infty}{\longrightarrow} 0
\end{equation*}
as a consequence of $\volSiglam = \lambda^2 \volunitdos + O(\lambda)$.
Concerning the term
in $\Ricg(k,k)$, we note that inserting the expansion  (\ref{thetakdev})
into the Raychaudhuri equation (\ref{Ray})  yields
$\Pi^{k}_{AB} \Pi^{k}{}^{AB} +  \Ricg(k,k) = O(\lambda^{-4})$
which implies $\Ricg(k,k) = O(\lambda^{-4})$ and again
$\int_{\Siglam} \Ricg(k,k) \volSiglam
\stackrel{\lambda \rightarrow \infty}{\longrightarrow} 0$.
Finally, $\theta_k^2 \volSiglam = (4  + o(1) ) \volunitdos$ from which
\begin{equation*}
\int_{\Siglam}\left( -\frac{1}{2}{\theta_k}^2 \right)\volSiglam+8\pi
\stackrel{\lambda \rightarrow \infty}{\longrightarrow} 0.
\end{equation*}
We next compute the derivative of $H(\lambda)$. The extrinsic curvature
$K^k$ along a null hypersurface satisfies the Ricatti equation
\cite{Galloway}
\begin{equation}
\label{Ricattiequ}
\frac{d}{d\lambda}(K^k)^A_{\phantom{A}B}=(K^k)^A_{\phantom{A}C}(K^k)^C_{\phantom{C}B}+R^A_{\phantom{A}B}.
\end{equation}
The trace-free part of this equation is
\begin{equation*}
\frac{d}{d\lambda}{\Pi^k}^A_{\phantom{A}B}=\theta_k {\Pi^k}^A_{\phantom{A}B}+R^A_{\phantom{A}B}-\frac{1}{2}\Ricg(k,k)\del^A_{\phantom{A}B},
\end{equation*}
where we have used 
$\Pi^k_{AB}(\Pi^k)^B_{\phantom{B}C}=\frac{1}{2}\tr((\Pi^k)^2)\gamma_{AC}$, which is an algebraic property
of endomorphisms in two-dimensional vector spaces. Thus
\begin{equation*}
\frac{d}{d\lambda}\left({\Pi^k}^A_{\phantom{A}B} {\Pi^k}^B_{\phantom{B}A} \right)=2\theta_k \tr((\Pi^k)^2) + 2 R^{AB}\Pi^k_{AB}.
\end{equation*}  
Using this together with (\ref{difvolform})  
and the Raychaudhuri equation, the derivative (\ref{derivh}) is obtained
after a number of cancellations.
\end{proof}

We can combine the previous computations to find a set
of sufficient conditions under which the 
renormalized area method applies.

\begin{Tma}[\bf Sufficient conditions for the
renormalized area method]
\label{penroseineqtma}
Let $\Omega$ be a strong past asymptotically flat null hypersurface 
and $\{\Siglam\}$ a geodesic foliation approaching large spheres. Assume that the spacetime satisfies the dominant energy condition.
If the two conditions
\begin{itemize}
\item[(i)] $ 
\left( \int_{\esfdos}\c\volunitdos \right)^2-8\pi\int_{\esfdos}\left(\c\right)^2\volunitdos-8\pi\int_{\esfdos}\cdos\volunitdos \geq 0$
\item[(ii)] $
\int_{\Siglam}\left( -2\theta_k\Ricg(k,k)+2(\Pi^k)^{AB}R_{AB}+\frac{d}{d\lambda}\Ricg(k,k)\right)\volSiglam\leq 0, \quad \quad \forall \lambda \geq 0$
\end{itemize}
hold, then
\begin{equation}
\sqrt{\frac{|S_0|}{16 \pi}} - \frac{1}{16 \pi} 
\int_{S_0} \theta_{\ell} \,
\bm{\eta_{S_0}} \leq E_B
\label{PI}
\end{equation}
where $E_B$ is the Bondi energy associated to the foliation $\{ \Siglam\}$.
In particular, if $S_0$ is a weakly outer trapped surface then
the Penrose inequality $E_B \geq \sqrt{\frac{|S_0|}{16\pi}}$ holds.
\end{Tma}
\begin{proof}
From (ii)  we have $H'(\lambda) \leq 0$  which implies $H(\lambda) \geq 0$,
as this function tends to zero at infinity. Hence $G'(\lambda) \leq0$. From (i)
and Proposition \ref{Prop5} we have $\underset{\lambda \to \infty}{\lim}G(\lambda) \geq 0$ and we conclude $G(\lambda) \geq 0$, or
equivalently $D'(\Siglam,\ell) \geq 0$. The theorem
follows from (\ref{mainineq}) using the fact that $\{\Siglam\}$
approaches large spheres.
\end{proof}

It is remarkable that $H'(\lambda)$ only involves curvature terms. This makes
checking the validity of $H'(\lambda) \geq 0$ feasible, at least in some cases.
In the next two sections we explore the validity of conditions (i) and 
(ii) in two simple, but relevant  situations.

\section{Shear-free vacuum case}
\label{Shear-free}

In this section we consider whether the
functional $M(\Siglam,\ell)$ can be used to
prove the Penrose inequality
in the case of 
{\it shear-free} null hypersurfaces $\N$ (i.e.
satisfying
$K^k = \frac{1}{2} \thk \gamma$)
embedded in a vacuum spacetime. The Penrose inequality
in this setup was proven
by Sauter \cite{Sauter2008} in full
generality exploiting properties of the Hawking
energy. Our interest in analyzing the shear-free case 
is to gain insight
on the range of applicability and limitations of the methods
discussed above.

For instance, concerning the renormalized area
method in subsection \ref{RAM}, the vacuum and shear-free
conditions immediately imply that
$H'(\lambda)=0$, so condition (ii) in Theorem \ref{penroseineqtma}
is always satisfied. Thus $H(\lambda)$ vanishes identically,
which is equivalent to $G(\lambda)=\mbox{const}$. The method
works if and only if this constant is non-negative. It can be computed
from its limit at infinity in Proposition \ref{Prop5} as
\begin{align}
G(\lambda) = \underset{\lambda \to \infty}{\lim}G(\lambda)=
\left( \int_{\esfdos}\c\volunitdos \right)^2-8\pi\int_{\esfdos}\left(\c\right)^2\volunitdos-8\pi\int_{\esfdos}\cdos\volunitdos.
\label{SFcase}
\end{align}
In the shear-free vacuum case, the 
Raychaudhuri equation (\ref{Ray})
is simply $\frac{d \thk}{d \lambda} = - \frac{1}{2} \thk^2$,
which integrates to
\begin{align*}
\thk = - \frac{2}{\lambda + \alpha},
\end{align*}
where $\alpha>0$ (because $\thk <0$ all along $\N$) is a Lie constant
function. Expanding near infinity
\begin{align*}
\thk = - \frac{2}{\lambda} + \frac{2 \alpha}{\lambda}
- \frac{2 \alpha^2}{\lambda^2} + O (\lambda^{-3}) 
\quad \quad \Longrightarrow \quad \quad \thkone = 2 \alpha, \quad \thktwo = -2 \alpha^2,
\end{align*}
which inserted into (\ref{SFcase})  yields
\begin{align*}
\G(\lambda) =  4 \left (  \left ( \int_{\esfdos} \alpha \volunitdos 
\right )^2               -  4 \pi \int_{\esfdos} \alpha^2 \volunitdos \right ).
\end{align*}
By the H\"older inequality this constant is always non-positive and vanishes
only when $\alpha = \mbox{const}$ (i.e. when $\{\Siglam\}$ is a GAB foliation). Except
in this case (which corresponds in the present setup to
$\thk |_{S_0} = \mbox{const}$) we have $G(\lambda) <0$ and 
$\D(\Siglam,\ell)$ is strictly monotonically decreasing, which makes the
renormalized area method method fail. In fact, as discussed 
in Remark \ref{EquivInfty},
the function $F(\Siglam)$ is also monotonically
decreasing, at least in a neighbourhood of infinity, so the approach
discussed in Proposition \ref{Prop2} also fails in the present setup.

Despite all this, 
the method involving the functional
$M(\Siglam,\ell)$ {\it is capable} of establishing the 
Penrose inequality in the shear-free vacuum case. However, as
we shall see next, the argument is not based on 
the monotonicity of $M(\Siglam,\ell)$ (which fails
in general, see below) but via an integration of (\ref{Var2}), which in
turn relies on the fact that 
all the geometric information  along
$\Omega$ can be computed explicitly in the shear-free vacuum case. 
From the shear-free condition and the expression for $\thk$,
the metric $\gamSiglam$ can be obtained from (\ref{Evolmetric})
\begin{align*}
\frac{d \gamSiglam}{d \lambda} = -2 K^k = - \thk \gamSiglam
= \frac{2}{\lambda + \alpha} \gamSiglam \quad \quad
\Longleftrightarrow \quad \quad 
\gamSiglam = (\lambda + \alpha)^2 \q
\end{align*}
where we used the fact that the foliation $\{ \Siglam\}$
approaches large spheres. The volume form is 
$\volSiglam = (\lambda + \alpha)^2 
\volunitdos$. 
As shown in Lemma \ref{monot}, the derivative of $M_b(\Siglam,\ell)$ involves
the connection one-form $\s_{\ell}$. This object satisfies
the following well-known evolution equation along an
arbitrary foliation defined by a null generator $k$ (see e.g.
\cite{MarsSoria2015})
\begin{align*}
k(s_{\ell}(X)) = - X(Q_k) - s_{\ell}(X) \thk +
(\mbox{div}_{S_r} K^k) (X) - D_X \thk
- \Eing(k,X)
\end{align*}
where $X$ is tangent to $\Siglam$ and satisfies $[k,X] =0$. In the 
vacuum, geodesic and shear-free case this equation becomes
\begin{align*}
\frac{d \s_{\ell}(X)}{d \lambda} = - \frac{2}{\lambda + \alpha} \s_{\ell}(X)
+ \frac{1}{(\lambda + \alpha)^2} X(\alpha)
\end{align*}
after using the explicit form of $\thk$. This equation can be integrated
to
\begin{align}
\s_{\ell} = \frac{1}{(\lambda + \alpha)^2} \left ( \lambda d \alpha + 
\bm{\omega} \right )
\label{sell}
\end{align}
where $\bm{\omega}$ is a Lie constant transversal one-form. 
In order to investigate the monotonicity
of the functional $\mmlam$ we need to evaluate (\ref{Var2}) and
in particular  $|\s_{\ell}|^2_{\gamSiglam} \volSiglam$. Using
(\ref{sell})
and the form of $\gamSiglam$  we have
\begin{align*}
|\s_{\ell}|^2_{\gamSiglam} \volSiglam =
\frac{1}{(\lambda + \alpha)^4} |\lambda d \alpha+ \bm{\omega}|^2_{\q} \volunitdos
\end{align*}
and identity (\ref{Var2}) simplifies to
\begin{align*}
\frac{d \mmlam}{d \lambda} 
=  \frac{1}{8 \pi} \sqrt{ \frac{4 \pi}{\int_{\esfdos} (\lambda + \alpha)^2 \volunitdos}}
\int_{\esfdos} (\lambda + \alpha ) \volunitdos  - \frac{1}{2}
+  \frac{1}{8 \pi} 
\int_{\esfdos} \frac{1}{(\lambda + \alpha)^4} |\lambda d \alpha+ \bm{\omega}|^2_{\q}
\volunitdos.
\end{align*}
We want to bound this expression from below. The Lie constant
one-form $\bm{\omega}$ can be uniquely split into
\begin{align*}
\bm{\omega} = - \beta d \alpha + \bm{\omega}^{\perp}, \quad 
\la \bm{\omega}^{\perp}, d \alpha \ra_{\q} =0
\end{align*}
where $\beta$ is a Lie constant function on $\Omega$. Thus
\begin{align}
\frac{d \mmlam}{d \lambda} 
& =  
\frac{1}{2} \left ( 
\frac{1}{\sqrt{4 \pi \int_{\esfdos} (\lambda + \alpha)^2 \volunitdos}}
\int_{\esfdos} (\lambda + \alpha ) \volunitdos  - 1 \right )
+  \frac{1}{8 \pi} 
\int_{\esfdos} \frac{((\lambda - \beta)^2 |d\alpha|^2_{\q} + |\bm{\omega}^{\perp}|^2_{\q})}{(\lambda + \alpha)^4} 
\volunitdos. \label{dermmshearfree} 
\end{align}
The H\"older inequality implies that the term in parenthesis is non-positive
and strictly negative unless $\alpha$ constant
(which corresponds both to the GAB case and 
also to the $D'(\Siglam,\ell) \geq 0$ case in the present context).
Since $\beta$ may be positive
and constant and $\bm{\omega}^{\perp}$ is allowed to be zero,  it follows that
$\frac{d \mm}{d \lambda} |_{\lambda = \beta}$ may have either sign.
This shows that one cannot expect $\mmlam$ to be
a monotonic functional on all cases. Nevertheless, the right-hand side
in (\ref{dermmshearfree})  is an explicit function in $\lambda$ that
can be integrated explicitly
\begin{align*}
M(S_{\lambda_1},\ell)   - M(S_0,\ell) =  & \left . \left [  
\frac{1}{2} \left ( \sqrt{\frac{\int_{\esfdos} (\lambda + \alpha)^2 \volunitdos}{4 \pi}}
- \lambda \right ) \right . \right . \\
& \left . \left . 
+ \frac{1}{8 \pi}
\int_{\esfdos} \frac{\left [ - \frac{\alpha^2}{4} - \frac{1}{3} 
\left (\beta - \frac{\alpha}{2}  \right )^2    
+ \lambda (\beta - \alpha)- \lambda^2 \right ] |d\alpha|^2_{\q} - 
\frac{1}{3} |\bm{\omega}^{\perp}|^2_{\q} }{(\lambda + \alpha)^3} 
\volunitdos \right ] \right |^{\lambda_1}_{\lambda=0}.
\end{align*}
Sending $\lambda_1$ to infinity, evaluating at $\lambda=0$
and using that the flow approaches large
spheres
\begin{align}
E_B  =  & M(S_0, \ell) 
+ \frac{1}{8 \pi} \left ( \int_{\esfdos} \alpha \volunitdos
- \sqrt{ 4 \pi \int_{\esfdos}  \alpha^2 \volunitdos}
+ \int_{\esfdos} \left( \frac{|d\alpha|^2_\q}{4 \alpha}  + \frac{(\beta 
- \frac{\alpha}{2} )^2 |d\alpha|^2_\q + |\bm{\omega}^{\perp}|^2_{\q}
}{3 \alpha^3} \right ) \volunitdos \right ) 
\nonumber \\
 = & \sqrt{\frac{|S_0|}{16 \pi}} - \frac{1}{16 \pi} \int_{\Sig_0} \thl
\volSigzero \nonumber \\
& + \frac{1}{8 \pi} \bigg (  
\underbrace{
\int_{\esfdos} \left( \alpha + \frac{|d\alpha|^2_{\q}}{4 \alpha} 
\right ) \volunitdos - 
\sqrt{ 4 \pi \int_{\esfdos}  \alpha^2 \volunitdos}}_{\defi I_1}
+ \underbrace{\int_{\esfdos}
\frac{(\beta - \frac{\alpha}{2} )^2 |d\alpha|^2_\q  +
|\bm{\omega}^{\perp}|^2_{\q}}{3 \alpha^3}
\volunitdos}_{\defi I_2} \bigg ).  \label{PIShearFree}
\end{align}
This identity is valid for any spacelike cross
section $S_0$ embedded in a shear-free and vacuum $\N$. 
We now use a fundamental identity for arbitrary $C^1$
functions $F$ on $\esfdos$, known as
the Beckner inequality \cite{Beckner1993}, which reads 
\begin{equation*}
\int_{\esfdos} \left ( F^2 + |d F|_{\q}^2 \right ) \volunitdos 
\geq \sqrt{ 4 \pi \int_{\esfdos}  F^4 \volunitdos}
\end{equation*}
with equality only for the constant functions. Writing $F = \sqrt{\alpha}$ it follows
\begin{align*}
\int_{\esfdos} \left ( \alpha + \frac{ |d\alpha|^2_{\q}}{4\alpha} \right ) \volunitdos
\geq \sqrt{ 4 \pi \int_{\esfdos}  \alpha^2 \volunitdos}
\end{align*}
and $I_1$ is non-negative.  The Penrose inequality in this case follows
because $I_2$ is manifestly non-negative and on a weakly outer trapped
surface $\thl \leq 0$.

The proof by Sauter \cite{Sauter2008} of this inequality in the vacuum, shear-free
case involved computing the Hawking energy
for a foliation $\{ S_s \}$ with the property $\theta_k (S_s) = \frac{2}{s}$.
This is in general a different foliation to the one used before 
(they only agree when $\alpha$ is constant). A fundamental 
step in Sauter's argument was also the Beckner
inequality. 
Note also,
that the Penrose inequality in the shear-free case involves not only the gap given by the Beckner inequality, but a second gap given by $I_2$.
The stronger Penrose inquality (\ref{PIShearFree}) is obviously  sharp 
because if $S_0$ is a MOTS ($\theta_{\ell}=0$)
we have equality in (\ref{PIShearFree}).
It is an interesting question whether
one can give a physical
interpretation to each of the two positive terms in (\ref{PIShearFree}).
Note that 
\begin{align*}
\bm{\omega} = \alpha^2 \s_{\ell} |_{\Sig_0} , \quad \quad \quad
\alpha = - \frac{2}{\thk |_{S_0}}, \quad \quad \quad
\q = \frac{1}{\alpha^2} \gamSigzero
\end{align*}
so that $\beta$ and
$\bm{\omega}^{\perp}$ can be determined in terms of the data
on $S_0$ and both $I_1$ and $I_2$ can be written fully in terms of the
geometry of the initial surface.

\section{Renormalized Area Method for the shell-Penrose inequality in $\Mcal$ }
\label{RAMMink}

The original setup where the Penrose inequality was
conjectured \cite{Penrose1973} involved 
an incoming null shell of dust matter propagating in the Minkowski
spacetime. By exploiting the junction conditions across the shell,
the Penrose inequality becomes a geometric inequality for
surfaces in the Minkowski spacetime. More precisely, if $S_0$
is a closed, connected, spacelike surface embedded in
the Minkowski spacetime ans satisfying a suitable convexity condition (which 
corresponds to the condition that its outgoing past null cone extends smoothly to
infinity), then 
\begin{equation}
\label{B2}
\int_{S_0} \theta_{\ell} {\bm{\eta_{S_0}}} 
\geq \sqrt{ 16 \pi |S_0|}
\end{equation}
where $\ell$ is the future directed null normal transverse to
the $S_0$ satisfying $\la \ell,k
\ra = -2$ and $k$ is future  null, tangent to the
outgoing past null cone generated by $S_0$ and normalized
by $\la k ,\xi\ra= -1$, where $\xi$ is a unit generator
of time translations. We call this conjecture the shell-Penrose inequality
in Minkowski (it has
also been called {\it Gibbons-Penrose inequality} in the literature). 
Analyzing the validity of this inequality is much simpler
than the general null Penrose inequality but it is still a challenging problem
which has received considerable attention in the literature \cite{BrendleWang,Gibbons1973,Gibbons1997,Mars2009,MarsSoria2012,MarsSoria2013,Sauter2008,Tod1985}.

It is a natural question to try and apply the general results
concerning the null Penrose inequality discussed above, to
the shell-Penrose inequality (\ref{B2}) in the
Minkowski spacetime. In this section we consider
the renormalized area method and in Section \ref{GABMink} we study the
GAB foliation.

The renormalized area method is particularly well-suited 
to the Minkowski spacetime. Indeed, the curvature tensor
vanishes identically in this spacetime, so  from Proposition
\ref{Prop6} we have that $H(\lambda)$ is constant and hence zero, 
as its limit at infinity always vanishes. Thus, as in the shear-free
case, $G(\lambda)$ is constant and its sign can be decided by its
asymptotic value (\ref{glimit}). We need
to determine $\c$ and $\cdos$. In the Minkowski spacetime this is
simple because $R^A_{\phantom{A}B}=0$ makes the Ricatti
equation explicitly integrable. 

As it is well-known (and easy to verify), the 
solution of the metric evolution equation
(\ref{Evolmetric}) and the Ricatti equation (\ref{Ricattiequ}) 
in Minkowski is given by

\begin{eqnarray}
\label{RicattiSol1}
& &(K^k)^{A}_{\,\,\,B} \Big{|}_p =  
\left . (\Kkzero)^{A}_{\,\,\,C} \right |_{\pi(p)}  [(\Id - \lambda (p)  \bm{\Kkzero} |_{\pi(p)})^{-1}]^{C}_{\,\,\,B}   \\
\label{RicattiSol2}
& &(\gamma )_{AB} \Big{|}_p = \left . (\gamma )_{AC} \right |_{\pi(p)} [(\Id - \lambda(p) \bm{\Kkzero} |_{\pi(p)})^2]^{C}_{\,\,\,B},
\end{eqnarray}
where  $\pi(p)$ is the (unique) point on $S_0$ lying on the null geodesic
containing $p$ and tangent to $k|_p$. Here $\bm{\Kkzero}$ denotes the
endomorphism with components $(\Kkzero)^{A}_{\,\,\,B}$
and $\Kkzero{}_{AB}$ stands to the null second fundamental
form of $S_0$ along $k$. Taking the trace of (\ref{RicattiSol1}) we find
$\theta_k |_p   =   (\Kkzero)^{A}_{\,\,\,C}
[(\Id-\lambda \bm{\Kkzero})^{-1}]^{C}_{\,\,\,A}   |_{\pi(p)}$, which for the sake
of simplicity we write simply as
\begin{eqnarray*}
\theta_k(\lambda) =   \tr \left [ \bm{\Kkzero} \circ \left ( 
\Id - \lambda \bm{\Kkzero} \right )^{-1} \right ],  \label{thetak}
\end{eqnarray*}
dropping all reference to the point $p$. This expression can be
immediately expanded near infinity to give
\begin{eqnarray*}
\theta_k=\frac{-2}{\lambda}+\frac{-\tr\left((\bm{\Kkzero})^{-1} \right)}{\lambda^2}+\frac{-\tr\left(  (\bm{\Kkzero})^{-2} \right)}{\lambda^3}+o(\lambda^{-3}).
\end{eqnarray*}
Thus,
\begin{equation}
\label{udefinition}
\c=-\tr\left( ( \bm{\Kkzero} )^{-1} \right):=u, \quad \quad
\cdos=-\tr\left(  (\bm{\Kkzero})^{-2} \right).
\end{equation}
Any  $2\times2$ matrix $\bm{A}$ satisfies
\begin{equation*}
\tr(\bm{A}^2)=\tr(\bm{A})^2-2\det(\bm{A}),
\end{equation*}
which applied to $\bm{\Kkzero}$ gives
 $\cdos=2\det\left ( (\bm{\Kkzero})^{-1} \right)-u^2$. Inserting this into
(\ref{glimit}) yields
\begin{equation}
\label{GlimitMink}
F_{\infty} = \underset{\lambda \to \infty}{\lim}G(\lambda)=\left(\int_{\esfdos}u\volunitdos\right)^2-16 \pi\int_{\esfdos}\left(
\det\left( (\bm{\Kkzero})^{-1} \right) \right)\volunitdos.
\end{equation}
This expression can be related to the area of $|S_0|$ as follows. 
From the definition
\begin{equation}
\q=\underset{\lambda \to \infty}{\lim} \frac{\gamma(\lambda)}{\lambda^2}= \underset{\lambda \to \infty}{\lim} \frac{\gamma(\Id-\lambda\bm{\Kkzero})^2}{\lambda^2}=\gamma(\bm{\Kkzero})^2,
\label{round}
\end{equation}
we can relate the volume forms at $S_0$ and ``at infinity'' by
\begin{equation}
\label{volformrelationesf}
\bm{\eta_{S_0}}=\det((\bm{\Kkzero})^{-1})\volunitdos
\end{equation}
and (\ref{GlimitMink}) becomes 
\begin{equation*}
F_{\infty} = \underset{\lambda \to \infty}{\lim}G(\lambda)
=\left(\int_{\esfdos}u\volunitdos\right)^2-16\pi|S_0|.
\end{equation*}
Summarizing, in the Minkowski spacetime $G(\lambda) = F_{\infty}$
and $F_{\infty} \geq 0$ implies (cf. Theorem \ref{penroseineqtma}) 
the validity of (\ref{PI}), which 
is exactly (\ref{B2}) because the Bondi energy of the Minkowski spacetime
vanishes identically. We have thus proved that the shell-Penrose inequality
in Minkowski holds provided
\begin{equation*}
\left(\int_{\esfdos}u\volunitdos\right)^2 \geq 16\pi|S_0|.
\end{equation*} 
In terms of the support function $h$ of $S_0$
(see 
\cite{MarsSoria2012} for its definition in the present context), this inequality
can be rewritten (after some manipulations) in the form
\begin{equation*}
4\pi\int_{\mathbb{S}^2}\big{(}(\triangle_{\q} h)^2+2h\triangle_{\q}
h\big{)}\volunitdos\geq
4\pi\int_{\mathbb{S}^2}u^2\volunitdos- \left (\int_{\mathbb{S}^2}u\volunitdos \right  )^2  ,
\end{equation*}
which is precisely the sufficient condition for the shell-Penrose
inequality in Minkowski obtained in \cite{MarsSoria2012}. This is not
surprising since the method in \cite{MarsSoria2012} also involved
a monotonicity condition for $\sqrt{\frac{|\Siglam|}{16 \pi}} - 
\frac{1}{2}\lambda$.
However, the general framework developed here leads to the result
in a much more efficient way. In fact, there is an even
more direct way of reaching this conclusion as a consequence of 
Proposition \ref{Prop2}, or rather of its rewriting in Remark
\ref{Prop2bis}. Indeed, from (\ref{RicattiSol2}) and
(\ref{volformrelationesf}),
\begin{align}
\volSiglam & =\det\left(  \bm{(\Kkzero)^{-1}}-\lambda \Id  \right)\volunitdos
=\left(\lambda^2+\c\lambda+\det \left ( \bm{\Kkzero})^{-1} \right) \right )
\volunitdos\quad \quad
\Longrightarrow \nonumber \\
|S_{\lambda} | & = 4 \pi \lambda^2 + 
\left ( \int_{\mathbb{S}^2} \c \volunitdos \right ) \lambda 
+ \int_{\mathbb{S}^2} \det \left ( \bm{\Kkzero})^{-1} \right) 
\volunitdos, \label{area}
\end{align}
where  in the second equality we used the first expression in
(\ref{udefinition}).
Comparing with (\ref{AreaExp}) it follows
that $\hat{\Theta}$  is Lie constant and takes the value
$\hat{\Theta} = |S_0|$, so that
the necessary condition (\ref{firstarticleineqgeneral})
becomes precisely $F_{\infty} \geq 0$.


The following proposition summarizes the results for the shell-Penrose
inequality in Minkowski obtained so far and shows, in addition, that
in the Minkowski case monotonicity of $D(\Siglam)$
is in fact equivalent to the a priori more general
conditions (\ref{problem}), or  $F'(\Siglam)\geq 0$.

\begin{Prop}[\bf Equivalence of the monotonicity methods in $\Mcal$]
Let $\Omega$ be a past asymptotically flat null hypersurface
in  $\Mcal$
and $\{ \Siglam\}$ a geodesic foliation approaching large spheres.
The following conditions are equivalent:
\begin{itemize}
\item[(i)] $\left(\int_{\esfdos}u\volunitdos \right)^2\geq 16\pi|S_0|$,
\item[(ii)] $\frac{d}{d\lambda}D(S_\lambda)\geq 0$ (Renormalized area method),
\item[(iii)] $\frac{d}{d\lambda}|S_\lambda|\geq\frac{16\pi|S_\lambda|}{8\pi\lambda+\int_{\esfdos}u\volunitdos}$ ($F'(\Siglam)\geq 0$ method),
\item[(iv)] $D(S_\lambda)\leq  \underset{\lambda \to \infty}{\lim} D(S_\lambda)$,
\end{itemize}
where $u=-\tr\left( ( \bm{\Kkzero})^{-1} \right)$. The shell-Penrose
inequality for $S_0$ holds if one (and hence any) of these
conditions holds.
\end{Prop}

\begin{proof}
The implications $(ii) \Longrightarrow (iii)$ 
and $(ii) \Longrightarrow (iv)$ are generally true. The equivalence
of $(i)$ and $(ii)$ is a consequence of $G(\lambda) = F_{\infty}$
and (\ref{GlimitMink}), as discussed above. We have also 
seen before that $(iii)$ is equivalent to $(i)$ 
as a consequence of Remark \ref{Prop2bis}. It only remains
to show that $(iv) \Longrightarrow (iii)$. Expression (\ref{area}) 
for the area $|\Siglam|$ yields
%
\begin{align*}
\frac{d}{d\lambda}|S_\lambda|\geq\frac{16\pi|S_\lambda|}{8\pi\lambda
+\int_{\esfdos}\c \volunitdos} 
\Longleftrightarrow  \left(8\pi\lambda+\int_{\esfdos} \c\volunitdos
\right)^2\geq 16\pi |S_\lambda| 
\Longleftrightarrow \sqrt{\frac{|S_\lambda|}{16\pi}}
-\frac{\lambda}{2}\leq \frac{1}{16\pi}\int_{\esfdos}\c\volunitdos,
\end{align*}
which establishes $(iv) \Longleftrightarrow (iii)$.

\end{proof}

\section{GAB foliations in $\Mcal$. Applications to the shell-Penrose inequality.}

\label{GABMink}
In the previous section we studied the renormalized area method
for the shell-Penrose inequality in Minkowski. In this section we
investigate in the same setting 
the consequences of the general Penrose-type inequality
obtained in Theorem  \ref{tmaprincipal}. To that aim we need
information on the limit of the Hawking energy along
GAB foliations. In \cite{MarsSoria2015} we have studied the limit
of the Hawking energy at infinity for a large class of foliations 
$\{ \Siglam \}$ along asymptotically flat 
null hypersurfaces. The results we need from that paper
can be summarized as follows:

Let $\{\Siglam\}$ be a geodesic background foliation approaching large
spheres and define $\c$, $\A$ and $\sone$ as in Definition \ref{AF}. 
Consider any other geodesic foliation $\{S_{\lambda'}\}$ starting on the same
cross-section $S_0$. The level-set functions $\lambda$ and 
$\lambda'$ are necessarily related by 
$\lambda=f\lambda'$, with $f>0$ and Lie constant on $\Omega$.
Then the limit of the Hawking energy along $\{ S_{\lambda'} \}$ is
\cite{MarsSoria2015}
\begin{equation}
\label{hawkingmasslimittwoone}
\underset{\lambda' \to \infty}{\lim}m_H(S_{\lambda'})=\frac{1}{8\pi\sqrt{16\pi}}\left(\sqrt{\int_{\mathbb{S}^2} f^2 \volunitdos}\right)\int_{\mathbb{S}^2}\left(\triangle_{\q}\c-(\c+\a)-4 \mathrm{div}_{\q} (\sone) \right) \frac{1}{f} \volunitdos.
\end{equation} 
In order to apply this result in the Minkowski context, we need
to compute $\c$, $\A$ and $\sone$ for the background foliation,
which we fix  as follows: choose a Minkowskian coordinate system $(t,x^i)$ and define 
the unit Killing $\xi=\partial_{t}$. The null generator $k$ of 
$\Omega$ is then uniquely selected by the condition
$\langle k,\xi\rangle=-1$ and $\{ \Siglam \}$  is defined to be
the level-set foliation
of $\lambda \in C^{\infty}(\Omega,\mathbb{R})$ defined by
$\lambda |_{S_0} =0$ and $k(\lambda)=-1$. It is immediate
to check that $\{ \Siglam\}$ approaches large spheres. 
%
%
%
%
The {\it time-height function} $\tau_{\lambda}$ of the level set
$S_{\lambda}$ is defined to be 
\begin{equation*}
\tau_\lambda:= t|_{S_\lambda}.
\end{equation*}
In particular $\tau_0=t|_{S_0}$ and, in fact, 
$\tau_\lambda |_p =\tau_0 |_{\pi(p)}-\lambda$ as a consequence 
of our choice of normalization for $k$. 


\begin{Lem}[\bf Asymptotic expansion at $\lambda=+\infty$]
\label{Lem7}
Let $\Omega$ be a past asymptotically flat null hypersurface
in $\Mcal$ and $\{ \Siglam \}$ a geodesic foliation associated
to a choice of Minkowskian coordinate system $\{ t, x^i \}$
as described above.
Let $\ell$ be orthogonal to $\{\Siglam\}$ and satisfying
 $\la \ell,k \ra = -2$. Then the following asymptotic
expansions hold
\begin{align}
\label{thetakexp}
\theta_{k}=&\frac{-2}{\lambda}+\frac{u}{\lambda^2}+o(\lambda^{-2}), \quad &
u& =-\tr\left( ( \bm{\Kkzero})^{-1} \right) \\
\label{thetaellexp}
\theta_{\ell}=&\frac{2}{\lambda}+\frac{-u+2\triangle_{\q}\tau_0}{\lambda^2}+ o(\lambda^{-2}) \quad &\tau_0 &:= t |_{S_0}\\
\label{sexp}
{s_{\ell}}_A=&\frac{-\esf_A\tau_0}{\lambda}+o(\lambda^{-1}).  &&
\end{align}
\end{Lem}

\begin{proof}
In the previous section we already proved (\ref{thetakexp}).
For $\theta_{\ell}$ we exploit the identity
\begin{equation}
\label{relcurvature}
\theta_{\ell}+(1+|\nabla\tau|^2_{\gamma})\theta_{k}-2\triangle_{\gamma}\tau=0,
\end{equation}
valid for any spacelike surface $S$ in Minkowski
whenever $\tau := t |_S$. This identity is a simple consequence
of the fact that $\xi$
is a covariantly constant vector field and it has been used several times
in the literature (we refer to 
\cite{MarsSoria2012} for a proof). We apply this identity to
$S_\lambda$ and expand for large $\lambda$ up to order
$\lambda^{-2}$. In particular, we can neglect all terms
of order $O(\lambda^{-3})$ or higher. Since 
$\tau_\lambda=\tau_0-\lambda$ and $\gamma_{S_{\lambda}}$ has the
expansion (\ref{invmetricexpansion}), the gradient term is 
 $|\nabla\tau|^2_{\gamma}= {\gamma_{S_\lambda}^{-1}}^{AB}{\tau_0}_{,A}{\tau_0}_{,B}
=O(\lambda^{-2})$ and the term $|\nabla\tau|^2_{\gamma}\theta_k$ is $O(\lambda^{-3})$
so that it can be ignored. Concerning the Laplacian term, since 
$\Delta_{\gamma} \tau = \Delta_{\gamma} \tau_0$ we have, in local coordinates
$\{ \lambda, y^A \}$ adapted to the foliation $\{ \Siglam \}$ (i.e.
such that $k = - \partial_{\lambda}$)
\begin{align*}
\triangle_{\gamma}\tau_{\lambda}
& =\frac{1}{\sqrt{\mathrm{det}(\gamma)}}\partial_A\left(\sqrt{\mathrm{det}(\gamma)}(\gamma^{-1})^{AB}\partial_B\tau \right)
\triangle_{\gamma}\tau_0 \\
& =\frac{1}{\sqrt{\det(\q)}}\partial_A\left(\sqrt{\det(\q)} (\q^{-1})^{AB} {\tau_0}_{,B}\right)\frac{1}{\lambda^2}+O(\lambda^{-3})=(\triangle_{\q}\tau_0)\frac{1}{\lambda^2}+O(\lambda^{-3})
\end{align*}
where we have used $\gamma(\lambda)=\q\lambda^2+O(\lambda)$. 
Inserting $\theta_{\ell} = \frac{2}{\lambda} + \frac{\a}{\lambda^2}+
o(\lambda^{-2})$ and (\ref{thetakexp}) 
into \eqref{relcurvature} and keeping only the
terms in $\lambda^{-2}$ we obtain
$\a + u - 2 \triangle_{\q}\tau_0 =0,$
which gives (\ref{thetaellexp}).

It only remains to compute $\sone$ in the expansion
$s_\ell=\frac{\sone}{\lambda}+o(\lambda^{-1})$. To that aim, we
decompose the Killing vector $\xi$ into normal and tangential
components to $\Siglam$ as
\begin{equation}
\xi=\frac{1}{2}\ell+\frac{(1+|D\tau_\lambda|^2_{\gamma_{\lambda}})}{2}k-
\mbox{grad} \, \tau_{\lambda}
\label{xikl}
\end{equation}  
where $\mbox{grad}$ is the gradient in $\Siglam$. This decomposition follows
directly from the normalization conditions and the definition of $\tau$
(an explicit derivation can be found in \cite{MarsSoria2012}, cf.
expression (16)). Solving for $\ell$ in (\ref{xikl}) and
inserting into the definition of $s_{\ell}$:
\begin{equation*}
{s_{\ell}}_A=\frac{1}{2}\langle\nabla_{X_A}k,\ell\rangle=\langle\nabla_{X_A}k, \xi-\frac{(1+|D\tau_\lambda|^2_{\gamma_{\lambda}})}{2}k+\mbox{grad} \tau_{\lambda}
\rangle= \langle \nabla_{X_A}k, \xi \rangle+\tau_\lambda^B K^k_{AB}.
\end{equation*}
Now, from $\langle k,\xi\rangle=-1$ we have
$\langle \nabla_{X_A}k, \xi \rangle = - \la k, \nabla_{X_A} \xi \ra =0$
because $\xi$ is covariantly constant. We conclude
\begin{equation*}
{s_{\ell}}_A=\tau_\lambda^B K^k_{AB},
\end{equation*}
from which the expansion (\ref{sexp}) follows directly after taking into account
(\ref{invmetricexpansion}) and (\ref{Kkexpansion}).

\end{proof}


Lemma \ref{Lem7} allows us to compute the limit of the Hawking energy along
very general foliations by exploiting the results in \cite{MarsSoria2015}.
For geodesic foliations $\lambda = f \lambda'$ we simply need to evaluate
(\ref{hawkingmasslimittwoone}), which becomes
\begin{equation}
\underset{\lambda' \to \infty}{\lim}m_H(S_{\lambda'})=\frac{1}{8\pi\sqrt{16\pi}}\left(\sqrt{\int_{\mathbb{S}^2} f^2 \volunitdos}\right)\int_{\mathbb{S}^2}\triangle_{\q}(u+2\tau_0) \frac{1}{f} \volunitdos.
\label{limitmhMink}
\end{equation} 
In particular, the GAB foliation associated to $S_0$ has rescaling
function $f\defi\frac{\c}{\ct}=\frac{u}{\ct}, \quad c>0$ so that,
along this GAB foliation, 
\begin{eqnarray*}
\underset{\lambda' \to \infty}{\lim}m_H(S_{\lambda'})&=&
\frac{1}{8\pi\sqrt{16\pi}}\left(\sqrt{\int_{\mathbb{S}^2} u^2 \volunitdos}\right)\int_{\mathbb{S}^2}\triangle_{\q}(u+2\tau_0) \frac{1}{u} \volunitdos.
\end{eqnarray*} 
Thus, the particularization of Theorem \ref{tmaprincipal} to the Minkowski setting
reads
\begin{Tma}
\label{IneqMink}
Let $\Omega$ be a past asymptotically flat null hypersurface 
in $\Mcal$ and $S_0$ a spacelike cross section of $\Omega$. Then the following inequality holds:
\begin{equation}
\label{Penroseineqtype}
\sqrt{\frac{|S_0|}{16\pi}} \leq  \frac{1}{16\pi} \int_{S_0} \theta_{\ell}\bm{\eta_{S_0}}+ \frac{1}{8\pi\sqrt{16\pi}}\left(\sqrt{\int_{\mathbb{S}^2} u^2 \volunitdos}\right)\int_{\mathbb{S}^2}\triangle_{\q}(u+2\tau_0) \frac{1}{u} \volunitdos ,
\end{equation}
where $u=-\tr\left( ( \bm{\Kkzero})^{-1} \right)$, $\tau_0 = t |_{S_0}$
with $t$ a Minkowskian time coordinate. The round asymptotic metric
 $\q$ is defined by (\ref{round}) and $\{ k, \ell\}$ are the future directed
null normals to $S_0$ with $k$ tangent to $\Omega$ and
satisfying $k(t)= 1$ and $\la k,\ell\ra = -2$.
\end{Tma}

\begin{remark}
Whenever $\underset{\lambda' \to \infty}{\lim} m_H(S_{\lambda'})\leq 0$, 
the shell-Penrose inequality in $\Mcal$ (\ref{B2})
for $S_0$ follows.
\end{remark}

As we discussed in Section \ref{Shear-free}, the Penrose inequality in the shear-free
case relies on a highly non-trivial Sobolev type
inequality for functions on the
sphere due to Beckner \cite{Beckner1993}. This inequality plays a core role
both in the proof by Sauter \cite{Sauter2008} and in the
proof presented in Section \ref{Shear-free}. The shear-free case in Minkowski
corresponds to the case where $\Omega$ is the past null cone of a point.
In fact, the shell-Penrose inequality for cross section $S_0$ on such a 
past null cone was first proven by Tod \cite{Tod1985} using a
Sobolev inequality on Euclidean space applied to suitable radially
symmetric functions. One might think that Sobolev type inequalities
of some sort should lie behind any method of proving the
shell-Penrose inequality for surfaces lying in the
past null cone of a point. We find it most remarkable that 
Theorem \ref{tmaprincipal} is capable of proving the 
shell-Penrose inequality with {\it no reference whatsoever} to
any Sobolev type inequality.

\begin{Cor}[\bf Shell-Penrose inequality in $\Mcal$ with spherical symmetry]
Consider a point $p\in\Mcal$ and $\Omega_p$ the past null cone of $p$. Let $S_0$ be a closed spacelike surface embedded in $\Omega_p$. Then the shell-Penrose inequality for $S_0$ holds true as a consequence of Theorem \ref{tmaprincipal}.  
\end{Cor}
\begin{proof}
The proof is immediate if we use the relation between $u$ and the support
function $h$, see \cite{MarsSoria2012}. We provide an 
alternative proof here for the sake of self-consistency.

Consider a Minkowskian time function $t$ and choose a value $t_0 < 
\inf_{S_0} \tau$, where $\tau = t |_{S_0}$.
Then the intersection of $\Omega_p$ with the hyperplane $\{t = t_0 \}$
is a round sphere $S_1$ of radius $t(p) - t_0$
and lying to the past to $S_0$. Let $k$ be
the null generator of $\Omega_p$ satisfying $k(t)=1$. Then the second
fundamental form along $k$ of $S_1$ is 
$(K^k)^{A}_{\,\,\,B} = - \frac{1}{t(p)-t_0} \delta^A_{\,\,\,B}$.
Since $K^k$ is both a property of $\Omega$ and of spacelike
surfaces embedded in $\Omega$ we conclude that  $(K^k|_q)^{A}_{\,\,\,B}=-\frac{1}{t(p)-t(q)}\delta^A_{\,\,\,B}$ 
for all $q\in \Omega$ and hence
$(\Kkzero)^A_{\,\,\,C} = - \frac{1}{t(p) - \tau_0} \delta^A_{\,\,\,B}$.
Thus $u = 2 (t(p) - \tau_0)$ which makes
the second term in the right-hand side of (\ref{Penroseineqtype})
identically zero.


\end{proof}

\begin{remark}
This argument proves, from (\ref{limitmhMink}), that 
the limit
of the Hawking energy at infinity on $\Omega_p$ vanishes for
all geodesic foliations. In fact, an explicit computation shows
that the Hawking energy
is {\it identically zero} for any cross section of $\Omega_p$.
\end{remark}


\section{An upper bound for the area of $S_\lambda$ 
along past asymptotically flat null hypersurface}

We close the paper returning to the general setup
of asymptotically flat null hypersurfaces in spacetimes satisfying the dominant energy condition. We also return to geodesic foliations not
necessarily approaching large spheres. In this section we provide
a general upper bound for the area $|\Siglam|$ in terms of asymptotic
quantities intrinsic to $\Omega$.
We find an inequality which is weaker than the inequality $D(\Siglam,\ell^\star)\leq \underset{\lambda \to \infty}{\lim}D(\Siglam,\ell^\star)$, the difference between both
being a H\"older inequality term. 

The general idea behind the inequality in the present section 
is the observation that one possible method to
approach the condition
$D(\Siglam,\ell^\star)\leq \underset{\lambda \to \infty}{\lim}D(\Siglam,\ell^\star)$
it to obtain an interpolating function $P(\lambda)$
satisfying $D(\Siglam,\ell^\star)\leq P(\lambda) 
\leq\underset{\lambda \to \infty}{\lim}D(\Siglam,\ell^\star)$. While this is hard 
(as finding such a $P(\lambda)$ would prove the Penrose inequality), we have 
been able to find a $P(\lambda)$ satisfying only the first inequality
$D(S_\lambda,\ell^\star) \leq P(\lambda)$, from which a general
inequality bounding $|S_0|$ from above in terms of asymptotic
quantities follows.
\begin{Prop}
Let $\Omega$ be a past asymptotically flat null hypersurface embedded in a spacetime that satisfies the dominant energy condition, $S_0$
a cross section and $\{\Siglam\}$ a geodesic foliation starting at $S_0$. Let $\c$ be the asymptotic coefficient
defined in (\ref{thetakdev}) and $\qh$ the asymptotic metric associated
to $\{ \Siglam \}$. Then,
\begin{equation}
|S_\lambda|\leq \frac{1}{4}\int_{\Sh}\left(\c+2\lambda\right)^2\volsh,
\label{BoungLast}
\end{equation}
and in particular $|S_0|  \leq \frac{1}{4}\int_{\Sh} (\c)^2 \volsh$.
\end{Prop}

\begin{proof}
Let us fix any $\lambda_0>0$ and consider the volume form on
$\Siglam$ ($\lambda \geq 0$) defined by
\begin{equation*}
\bm{\hat{\eta}_{\Siglam}}
\defi \frac{1}{(\lambda+\lambda_0)^2}\volSiglam.
\end{equation*} 
Using the evolution equation
$\frac{d}{d\lambda}\bm{\eta_{S_\lambda}}=-\theta_k\bm{\eta_{S_\lambda}}$, the derivative
of $\bm{\hat{\eta}_{\Siglam}}$
is 
\begin{equation}
\label{diffSiglamhat}
\frac{d}{d\lambda}(\bm{\hat{\eta}_{\Siglam}})=-\left(\theta_k+\frac{2}{\lambda+\lambda_0}\right)\bm{\hat{\eta}_{\Siglam}}.
\end{equation}
Writing $\bm{\hat{\eta}_{\Siglam}}=\hat{f}(\lambda)\volsh$,
(\ref{diffSiglamhat}) becomes a differential equation for $\hat{f}$,
which can be integrated as
\begin{equation*}
\hat{f}(\lambda)=\hat{f}(0)e^{-\int_0^\lambda\left(\theta_k+\frac{2}{s+\lambda_0} \right)ds}.
\end{equation*} 
The initial value $\hat{f}(0)$ can be computed ``at infinity'' as a consequence
of $\bm{\hat{\eta}_{\Siglam}} \longrightarrow \volsh$ when $\lambda \rightarrow
\infty$. Thus
$\hat{f}(0)=e^{\int_0^\infty\left(\theta_k+\frac{2}{s+\lambda_0} \right)ds}$ and therefore
\begin{equation*}
\hat{f}(\lambda)=e^{\int_\lambda^\infty\left(\theta_k+\frac{2}{s+\lambda_0} \right)ds}.
\end{equation*}
We aim at finding an upper bound for
$\hat{f}(\lambda)$. We use the inequality (\ref{sufdos}), which implies
 $\theta_k+\frac{2}{\lambda+\lambda_0}\leq\frac{-4}{\c+2\lambda}+\frac{2}{\lambda+\lambda_0}$ and then 
\begin{equation*}
\int_\lambda^\infty\left(\theta_k+\frac{2}{s+\lambda_0} \right)ds\leq \int_\lambda^\infty\left( \frac{-4}{2s+\c}+\frac{2}{s+\lambda_0} \right)ds=
\log\left( \frac{2\lambda+\c}{2 (\lambda+\lambda_0)}\right)^2 .
\end{equation*}
Finally,
\begin{eqnarray*}
 &&|S_\lambda|=\int_{\Siglam}\volSiglam=\int_{\Sh}(\lambda+\lambda_0)^2\hat{f}(\lambda)\volsh=\int_{\Sh}(\lambda+\lambda_0)^2e^{\int_\lambda^\infty\left(\theta_k+\frac{2}{s+\lambda_0} \right)ds}\volsh\leq  \\
&&(\lambda+\lambda_0)^2\int_{\Sh}e^{\log\left( \frac{2\lambda+\c}{2(\lambda+\lambda_0)}\right)^2 }\volsh=\frac{1}{4}\int_{\Sh}(\c+2\lambda)^2\volsh.
\end{eqnarray*}

\end{proof}

\begin{remark}
The condition $D(S_\lambda,\ell^\star)\leq \underset{\lambda \to \infty}{\lim}D(S_\lambda,\ell^\star) $, namely
\begin{equation*}
\sqrt{\frac{|\Siglam|}{16\pi}}-\frac{\rq}{2}\lambda \leq \frac{1}{16\pi\rq}\int_{\Sh}\c\volsh,
\end{equation*}
is equivalent to
\begin{equation*}
|S_\lambda|\leq \frac{1}{16\pi\rq^2} \left( \int_{\Sh}(\c+2\lambda)\volsh\right)^2.
\end{equation*}
As mentioned above, this inequality is stronger than (\ref{BoungLast}), the
difference being a H\"older inequality term. Indeed, a direct
application of the H\"older inequality yields
\begin{equation*}
|S_\lambda|\leq \frac{1}{16\pi\rq^2} \left( \int_{\Sh}(\c+2\lambda)\volsh\right)^2 \leq \frac{1}{4}\int_{\Sh}(\c+2\lambda)^2\volsh
\end{equation*}
\end{remark}



\section*{Acknowledgments}
Financial support under the project  FIS2012-30926 (MICINN)
is acknowledged. A.S. acknowledges the Ph.D. grant AP2009-0063 (MEC).

\end{document}